\documentclass[a4paper, 11pt]{article}
\usepackage{dsfont}

\usepackage{calrsfs}
\usepackage{comment}
\usepackage{ulem}

\usepackage{amsmath, amsthm, amssymb, amsfonts}
\usepackage{float} 
\usepackage[shortlabels]{enumitem}
\usepackage[square]{natbib}
\usepackage{float}

\usepackage[usenames]{color}
\usepackage{mathtools}
\usepackage{tikz}

\usetikzlibrary{arrows,decorations.pathmorphing,backgrounds,positioning,fit,automata}
\usetikzlibrary{petri}
\usetikzlibrary{topaths}
\usepackage{comment}

\usepackage[colorlinks=true,breaklinks=true,bookmarks=true,urlcolor=blue,
     citecolor=blue,linkcolor=blue,bookmarksopen=false,draft=false]{hyperref}
     
\newcommand{\ignore}[1]{}

\newtheorem{theorem}{Theorem}

\newtheorem{lemma}[theorem]{Lemma}

\newtheorem{proposition}[theorem]{Proposition}

\newtheorem{definition}[theorem]{Definition}

\newtheorem{remark}[theorem]{Remark}
\newtheorem{example}[theorem]{Example}
\DeclareMathOperator{\card}{Card}

\usepackage[square]{natbib}

\begin{document}

\title{Location games with references}
\author{
Ga\"etan Fournier\thanks{Aix Marseille Univ, CNRS, AMSE, Marseille, France, gaetan.fournier@univ-amu.fr. The project leading to this publication has received funding from the french government under the “France 2030” investment plan managed by the French National Research Agency (reference: ANR-17-EURE-0020) and from Excellence Initiative of Aix-Marseille University - A*MIDEX. "}~~ \& ~~ Amaury Francou\thanks{This research was initiated during the student research internship in AMSE of Amaury Francou, amauryfrancou@gmail.com}}

\maketitle

\begin{abstract}
We study a class of location games where players want to attract as many resources as possible and pay a cost when deviating from an exogenous reference location. This class of games includes political competitions between policy-interested parties and firms' costly horizontal differentiation. We provide a complete analysis of the duopoly competition: depending on the reference locations, we observe a unique equilibrium with or without differentiation, or no equilibrium. We extend the analysis to a competition between an arbitrary number of players and we show that there exists at most one equilibrium which has a strong property: only the two most-left and most-right players deviate from their reference locations. 
\end{abstract}



\section{Introduction}
Spatial competitions arise when sellers maximize their clientele by competing on the space of customers' preferences or when political parties maximize their electorate by competing in the space of political opinions, as introduced in the seminal works of \citep{hotelling1990stability} and \citep{downs1957economic}. In the most general version of these models, we can ignore the exact objective of the agents and simply assume that they seek to maximize a resource through spatial competition. 

However, economic agents rarely have complete freedom when choosing their locations: there is typically a reference situation, described by the characteristics or the past choices of the agents or their predecessors. These references are not necessarily controlled by the agents but have a significant impact on the outcome of the competition and can not be ignored.

A first example one can think of is the situation where vendors choose the location of their store. While the Hotelling-Downs model analyzes a competition between sellers entering a new market, it is sometimes more realistic to assume that they already have a historical address and that a move is costly. The price of the move has no impact on the consumers, whose economic behaviors depend only on the new locations, however it has a strategic impact on the sellers' decisions. 

The general situation we study is the following: a finite number of players compete for uniformly distributed resources on $[0,1]$. Each player $i$ has an exogenous reference $r_i$, selects a location $x_i$ and pays a cost that increases with the distance between $x_i$ and $r_i$. Resources are attracted by the player with the closest location.
 The model, formally detailed in section \ref{se:model}, can be applied in numerous fields of research. As an illustration we now mention two motivating examples.

\subsection{Motivating examples}\label{subse:2examples}
\textbf{\underline{Example 1:} Political competition with policy-motivations.\\} 
In the classical Downs model, political parties only have \textit{office motivations}: they simply consider a policy as an opportunistic instrument, to be successful in the election. In other words, they would say anything to seduce the electorate. Therefore, they compete on the policy space by selecting platforms $x_i \in [0,1]$ in order to maximize their electoral success. It is sometimes more realistic to suppose that political parties also have \textit{policy motivations}, for example when they do not want to betray their sincere favorite policy $r_i$. In this situation, parties have to compromise between two objectives: maximizing their vote shares\footnote{The vote share might also represent the probability to win a winner-takes-all election because voters have incomplete information about political parties, see \citep{lindbeck1987balanced} for example, or because parties have incomplete information about the distribution of voters' preferences \citep{patty2002equivalence}.} and minimizing the distance between their sincere opinions $r_i$ and the selected platforms $x_i$. Following \citep{roemer2009political}, each motivation represents the objective of a faction of political parties, the \textit{opportunists} and the \textit{militants}: while the \textit{opportunists} "desire to maximize the probability of the party's victory", the \textit{militants} "desire to propose a policy as close as possible to the party's ideal point and have little interest in winning the election".\\
~~\\
\textbf{\underline{Example 2:} Costly products' differentiation.\\} Consider a market where firms want to maximize their clientele by selling a product whose price is exogenously determined. This model applies for example to newsstands, pharmacies or franchises of different types of services and products. It also applies to the media markets or to any sector where consumers are not charged with direct prices but through advertising.

We suppose that the good has a $1$-dimensional characteristic and that consumers have symmetric single peaked preferences over this characteristic, with a maximum uniformly distributed on the interval, so that consumers choose to buy one unit of the product whose characteristic is the closest to their ideal. Each firm $i$ has the exogenous expertise to produce a good with characteristic $r_i$ and has the possibility to produce a good with characteristic $x_i \neq r_i$ for a certain cost which increases with the distance $d(r_i,x_i)$. The literature mentions this investment as the "differentiation costs" (see for example \citep{eaton1994flexible}). When the characteristic is geographic, it may be interpreted as the cost of transporting inputs bought from suppliers for example, otherwise it might be interpreted as resulting from the process of modifying the standard product of the firm, through research and development investments for example. The payoff of firms is again composed of two terms: the first is the quantity of product they sell, which is a good proxy for their profits as the price is fixed and because we suppose that the production costs are small enough to give the firms the incentives to sell as much as possible. The second term represents the cost of the differentiation.\\
~~\\
The literature review, in subsection \ref{subse:lit_review}, comes back to the above examples and illustrates our contributions to the literature. We also briefly mention other applications of our model.

\subsection{Main results}

We first study the duopoly competition. We provide a complete characterization of the equilibrium for any pair of exogenous references $(r_1,r_2) \in [0,1]^2$. An interesting question tackled by the literature is to explain in which context do we observe differentiation or not: the principle of minimal differentiation states that in many different spatial competitions we observe similar equilibrium locations for both players.
Although this result is quite robust\footnote{The principle still holds in the case of non-uniform distribution of consumers, in the case of sequential entry of the players, in a winner takes all competition, or when there is uncertainty about the preferences if players share the same prior.}, the principle has been challenged by empirical evidences and theoretically by the introduction of price competition (\citet{d1979hotelling}) or turnout in voting models (\citep{davis1970expository}) for example. In the current paper, we prove that an equilibrium with or without differentiation can be observed depending on the pair of references $(r_1,r_2)$. As stated formally in Proposition \ref{prop:2players}, there are three different scenarios:
\begin{enumerate}
\item If references are far from each other, there exists a unique equilibrium where players differentiate.
\item If references are close to each other but far from $\frac{1}{2}$, there is no equilibrium.
\item If references are both close to each other and close to $\frac{1}{2}$, there exists a unique equilibrium where players do not differentiate $(x_1^*=x_2^*=\frac12)$.
\end{enumerate}

The intuition is the following: each player is subject to two possibly conflicting forces, on the one hand each player wants to move towards his opponent as such a move always increases the quantity of attracted resources, but on the other hand they does not want to move too far from their reference location. When the references are far from each other, each of the two players can act as a local monopolist and find the optimal compromise between the two forces. However, when the references are close to each other, the first order conditions can not be satisfied without passing each other positions. If references are close to $\frac12$, this median location is chosen by both players at equilibrium and the minimal differentiation principle is verified: the unique equilibrium is $(\frac12,\frac12)$. Surprisingly, equilibrium locations are not necessarily located in the interval between the two references. Finally, if the references are close to each other but far from $\frac12$, the situation is unstable: players still want to be as close as possible to each other, but if they select the same location $x \neq \frac{1}{2}$ then each of them would benefit from an infinitesimal deviation towards $\frac12$ as he would attract more than half the resources.\\

Our results quantitatively depend on the relative importance of the cost of deviations from the references, which is represented by a parameter $c>0$. We provide comparative statics to show that the larger $c$ is, the more often we observe a differentiated equilibrium and the less often we observe the equilibrium $(\frac12,\frac12)$. All together, we show that the likeliness equilibrium existence is a non-monotonic function of the costs: first decreasing then increasing with $c$. We explain this phenomenon and illustrate graphically the different possible equilibria on the $2$-dimensional space of $(r_1,r_2) \in [0,1]^2$. Finally, we investigate the heterogeneous case where players have different parameters $c_1 \neq c_2$.\\

We then turn to the competition with more than two players. Among oligopolies, the triopoly has a particular interest as this configuration is singular in the literature of standard location games: it is the only configuration where there exists no equilibrium. We prove that this non-existence result is robust against the introduction of a small cost of deviation. We provide an explicit upper bound for this robustness. If the costs increase above this threshold, we prove that there exist three possible configurations depending on the triplet of references: no equilibrium, a unique equilibrium where all players differentiate or a unique equilibrium where two players do not differentiate.\\
~~\\
We then consider the general case of $n$ players\footnote{The general case is the one of $n \geq 5$ players. The less interesting case of four players slightly differs and is briefly studied aside.}. We prove that there exists either a unique or no equilibrium depending on the reference locations. This result differs from the spatial competition without references where there exists, for $n>5$, a continuum of equilibria (see for example \cite{eaton1975principle} and notice that the set is not very tractable as it involves $n^2$ equations). The introduction of the reference locations surprisingly simplifies the equilibrium structure as it creates a tie-breaking rule. We characterize the unique possible equilibrium and show a strong property: at most four players select locations which are different from their references, the two players with the most left or right references. A key argument for this property is very intuitive: when a player changes his location within the same interval between his neighbors, he gains some resources on one side but looses the same quantity on the other side, so only the cost minimization matters. We provide necessary and sufficient conditions under which the unique equilibrium candidate is indeed an equilibrium, and we illustrate our result in the simple case where references are regularly located on the unit interval.\\

Finally, we implement a simple algorithm that computes the equilibrium candidate for any references vector, and determine whether it is or not an equilibrium. To make it possible, we show that despite the continuous action space $[0,1]^n$, there only exist a finite number of possible best responses for each player. 

\subsection{Literature review}\label{subse:lit_review}
 
In political economy, many papers are concerned with spatial competition. While the original paper of \citep{downs1957economic} only considers office-motivated candidates, \citep{wittman1973parties} considers a competition between two policy-interested candidates and both papers gave birth to different literatures.\footnote{A few papers tackle the question of the candidates' motivations from an empirical point of view, see for example \citep{fredriksson2011politicians}.} Between these two extremes, some papers analyze the case of parties with hybrid motivations, for example \citep{callander2008political}, \citep{saporiti2008existence} or \citep{DROUVELIS201486}. The latter claims that \textit{"although politicians might be more interested in winning the elections, it seems reasonable to expect that policy considerations will also enter into the candidate's payoff function with some weight"}. The main difference with our paper is that the policy motivation of candidates is measured by the distance between the ideal policy of the parties and the policy implemented after the election. In the political interpretation of our model, we do not consider a winner-takes-all election and therefore, we do not specify which policy is implemented after the election. On the other hand, we consider that players are political parties that have to compromise between maximizing vote shares and minimizing the distance between the preferred policy and the platform proposed by the party.\\
In this sense we refer to \citep{roemer2009political}: the policy-interested candidates introduced in the above papers are defined as the \textit{reformists} while our policy-interested candidates are defined as the \textit{militants}. When the reformists minimize the distance between their ideal and the implemented policy, the militants minimize the distance between their ideal and the policy proposed by the party.\footnote{Roemer also mention that "political histories are replete with descriptions of these three kind of party activists (the third kind being the opportunists that are office-interested). For instance, Schorske (1993) calls them, when describing the German Social Democratic Party: the party bureaucrats, the trade union leadership and the radicals".} The solution concept studied by Roemer differs from ours: his "party-unanimity Nash equilibrium" (PUNE) is an equilibrium in the game where a deviation is profitable by the party only when it is unanimously preferred by all factions of the party. The author provides some existence results and gives a Nash bargaining interpretation of the solution concept. \citep{kartik2007signaling} also consider a competition where some candidate are \textit{militants} in the sense that they do not want to propose a platform they do not \textit{believe in}. In their model, such candidates are preferred by voters and office-motivated candidates use mixed strategies to pretend they also care about politics.\\

We now discuss papers concerned with the use of location games with references in industrial organization. Our contribution in this field is to consider differentiation costs paid by firms in order to produce a good which has a different characteristic than the standard good. \citep{kishihara2020product} consider a costly product re-positioning performed by two firms having predetermined product characteristic base positions. The two firms compete on both locations and prices, but the authors can only provide analytic results for the equilibrium outcomes, as closed-form solutions can not be given. \citep{correia2011costly} suppose that firms pay a quadratic cost to differentiate their products from the standard product which is fixed to $\frac{1}{2}$ for all agents. Firms compete on both location and price, and the authors found that there is either a differentiated equilibrium or no equilibrium at all. Although we do not take into account the impact of locations on prices, we extend this analysis to the case where the standard products are not necessarily fixed at $\frac12$ and can in particular be different. We also relax the restriction on the number of firms. We do not believe that our analysis could keep this generality in a location-then-price setting. Notice that \citep{eaton1994flexible} also suppose that firms produce goods which may be more or less costly to produce, according to whether their specifications are distant from or close to the specifications of a basic product. In a survey, \citep{brenner2010location} mentions that \textit{"a series of influential papers have abstracted away from the choice of the price, and looked at whether location equilibria exist, and how they can be characterized when the location has no impact on the pricing"}. We refer to this survey for the literature concerned with fixed price location games and to \citep{loertscher2011sequential} for a broader discussion on markets where the hypothesis of fixed price is particularly relevant. 
\\
We briefly mention papers concerned with other applications of our model. When a social planner prefers the firms to produce a particular good for exogenous reasons, firms are penalized for the distance between the good they produce and a subsidized good. In this case, firms want to compromise between taking strategic opportunities to increase their market shares and maximizing the subsidies by following the recommendation of the planner. In \citep{lambertini1997optimal}, the author examines a duopoly in which he considers that firms can be either taxed or subsidized depending on the choice of the characteristics, which influence the equilibrium locations. Finally, let us mention that in the literature of strategic advertising, when firms produce a good with a certain characteristics but pretend, through advertisements, to produce a different good in order to attract naive consumers, it is usually penalized by a cost function. Such a cost typically increase with the difference between the two products. For example, \citep{tremblay2002advertising} considers a duopoly competition where the firms produce an identical product located at $\frac12$, but can generate a subjective product differentiation through costly advertising.

\section{Model and notations}\label{se:model}

We now formally define the location game in its normal form. A finite number $n$ of players simultaneously select locations $\textbf{x}=(x_1,\dots,x_n)$ in the action space $[0,1]^n$. In order to explicit the players' payoffs, we need to define $q_i(\textbf{x})$, the quantity of resources that player $i$ attracts when the profile of locations $\textbf{x}=(x_1,\dots,x_n)$ is played. 
It is the quantity of resources $t$ that are closer to $x_i$ than to any other other $x_j$. In the case where several players select the same location, they split their resources equally. Only a negligible set of resources are at the same distance from two different locations. Formally:
$$q_i(x_1,\dots,x_n):=\displaystyle\frac{\mu(\{ t \in [0,1],~|t-x_i| \leq |t-x_j| \text{ for every } j \neq i \})}{Card(\{ x_j | x_j = x_i \})}$$
where $\mu$ is the Lebesgue measure (simply used to measure the lengths of intervals).\\

Players are endowed with exogenous references $\textbf{r}=(r_1,\dots,r_n)$ and each player $i$ incurs as cost when choosing a location $x_i$ which is different from his reference location $r_i$. This cost is given by $\gamma(|x_i-r_i|)$ where $\gamma: [0,1] \rightarrow \mathbb{R}$ is a continuous, strictly increasing and strictly convex function\footnote{The convexity of the cost with respect to the distance is a common assumption: see for example \citep{d1979hotelling}, where the author considers a the quadratic cost functions to investigate the differentiation in a location-then-price models and where linear costs induce technical difficulties or alternatively \citep{villani2009optimal} for the use of quadratic cost of distance in optimal transport.} such that $\gamma(0)=0$. We sometimes assume that $\gamma(d)=cd^2$ where $c>0$, in order to obtain explicit results and to provide a comparative static with respect to the parameter $c$. However, we show that the main results hold without this simplifying assumption.

Player $i$'s payoff is then given by:
$$g_i(\textbf{x},\textbf{r})= q_i(\textbf{x})-\gamma(|x_i-r_i|)$$

It is useful to define the \textit{left (resp. right) neighbor} of player $i$ as the player that selects the largest (smallest) location  which is strictly smaller (larger) than $x_i$, if any, and the \textit{neighborhood} of player $i$ as the interval $(x_i^{\ell},x_i^{r})$ where $x_i^{\ell}$ is the left neighbor of player $i$ if it exists and the boundary $0$ otherwise and $x_i^{r}$ is the right neighbor of player $i$ if it exists and the boundary $1$ otherwise. Finally, we say that a player is a \textit{left (resp. right) peripheral player} when he has no left (resp. right) neighbor.

\section{The duopoly}\label{se:duopoly}

In this section, we suppose for simplicity that $\gamma(d)=cd^2$ with $c>0$, except in subsection \ref{subse:general_costs} where we generalize the results.

\subsection{Existence and characterization of equilibrium}
\begin{proposition}\label{prop:2players}{\textbf{The duopoly competition}\\}
Let $\textbf{r}=(r_1,r_2) \in [0,1]^2$ and suppose, without loss of generality, that $r_1 \leq r_2$.
\begin{itemize}
\item If $r_2-\frac{1}{4c} \leq r_1+\frac{1}{4c}$, the unique possible equilibrium is $(\frac{1}{2},\frac{1}{2})$.\\~~\\ It is an equilibrium if and only if $\frac12 \in [r_2-\frac{1}{4c},r_1+\frac{1}{4c}]$.
\item If $r_1+\frac{1}{4c}<r_2-\frac{1}{4c}$, the unique possible equilibrium is $(r_1+\frac{1}{4c},r_2-\frac{1}{4c})$.\\~~\\ It is an equilibrium if and only if $\begin{cases} c(r_2-r_1)^2+r_1+r_2-1-\frac{1}{4c} \geq 0~~~~~~(E1)\\ c(r_2-r_1)^2-(r_1+r_2)+1-\frac{1}{4c} \geq 0~~~~(E2) \end{cases}$
\end{itemize}
\end{proposition}

The proof is postponed to subsection \ref{prof:prop2players} and we now illustrate and comment the above proposition.\\


\underline{Illustration and intuition:} The figure below describes the equilibrium structure with respect to the references $(r_1,r_2) \in [0,1]^2$. 


\begin{figure}[H]
\centering
\includegraphics[scale=0.2]{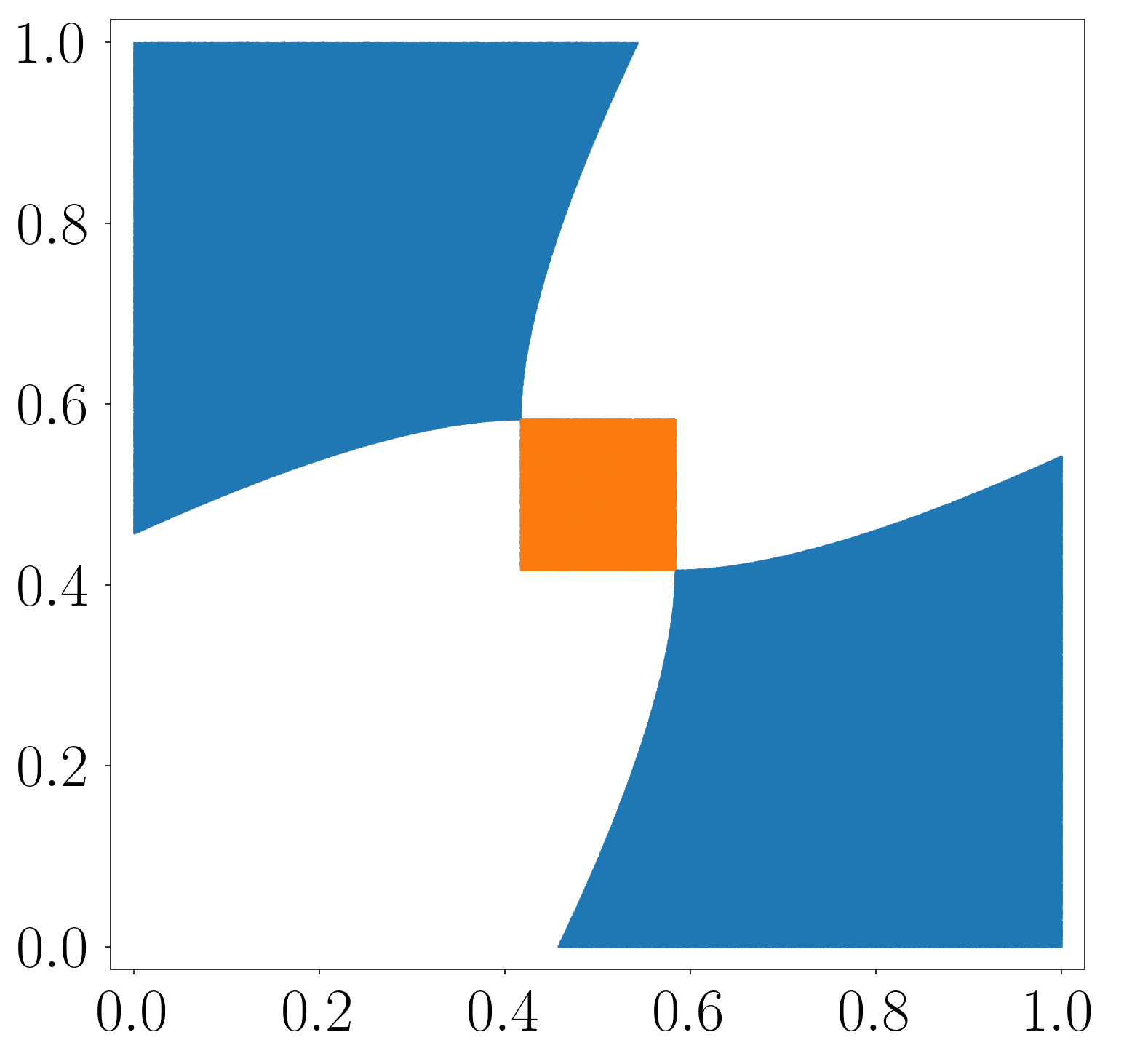}
\caption{The different equilibrium configurations with respect to $(r_1,r_2)$, for $c=3$. In white there is no equilibrium, in orange $(\frac12,\frac12)$ is the unique equilibrium, in blue $(r_1+\frac{1}{4c},r_2-\frac{1}{4c})$ is the unique equilibrium.}
\label{fi:duopoly}
\end{figure}

As discussed in the introduction, we observe $3$ possible scenarios:
\begin{enumerate}
\item In the orange square, references are such that the game admits a unique equilibrium which is $(\frac{1}{2},\frac{1}{2})$ . As stated in Proposition \ref{prop:2players}, this scenario happens when references are both close to each other: $|r_2-r_1| \leq \frac{1}{2c}$ and close to the center: $|r_2-\frac{1}{2}|  \leq \frac{1}{4c}$ and $|\frac{1}{2}-r_1| \leq \frac{1}{4c}$.\\
~~\\
The proof of Proposition \ref{prop:2players} combines two ingredients in this scenario. The first ingredient is concerned with how players compromise between maximizing resources and minimizing costs. On the one hand, when moving closer to his opponent, the marginal gain of resources is constant and equal to $\frac{1}{2}$, as the midpoint between the player and his neighbor moves two times slower than the player. On the other hand the marginal cost of a deviation further the references is increasing because the cost function $\gamma(d)=cd^2$ is convex. Therefore, we find that there exists an optimal deviation from the reference, which is given by a move of $\frac{1}{4c}$ towards the opponent (later named $\delta$ in the general case of a convex function $\gamma$), as long as they don't pass each other's positions. If the distance between the two references is closer than $\frac{1}{2c}$, then at equilibrium players necessarily select the same location, otherwise at least one player has a profitable deviation by moving closer to his opponent. The second ingredient is formalized in claim 2 of Lemma \ref{prop:3claims}: $(\frac{1}{2},\frac{1}{2})$ is the only possible equilibrium where players select the same location. Otherwise, each player has a profitable deviation by moving marginally towards the center. Obviously, if references are to far from $\frac{1}{2}$, the profile $(\frac12,\frac12)$ is not an equilibrium. Because the constraints are stronger when $c$ increases, this first scenario is less likely as the cost of deviating from the references becomes higher. 

In this configuration, a counter-intuitive phenomenon can happen: it might be that players select equilibrium locations that are strictly larger than both references (or smaller). Take for example $c=3$ and $r=(0.42,0.45)$, then the equilibrium is $(\frac12,\frac12)$. 

\item In the two symmetric blue areas, the game admits a unique equilibrium which is $(r_1+\frac{1}{4c},r_2-\frac{1}{4c})$. In this case, references are far enough from each other so that each agent acts as a local monopolist and deviates from his reference towards his opponent by $\frac{1}{4c}$.\\
~~\\
In this configuration, players do not have a profitable deviation within the same neighborhood, that is on the same side of the opponent. Therefore, the proof of Proposition \ref{prop:2players} focuses on the cases where players have a profitable deviation by relocating on the other side of their opponent.

Consider the situation where $r_1<x_1<x_2<r_2$. Player $1$'s payoff is decreasing with $x_1$ on the whole interval $(x_2,1]$ as it decreases the resources he attracts and increases his costs. Therefore, the equilibrium conditions are limited to one equation for Player $1$ not to deviate marginally on the right of Player $2$\footnote{We prove that $\displaystyle\lim_{\epsilon \downarrow 0} u_1(x_2+\epsilon,x_2) \leq u_1(x_1,x_2)$ and symmetrically for Player $2$.} and one other condition for Player $2$ not to deviate marginally to the left of Player $1$. These equations are denoted $(E1)$ and $(E2)$ in Proposition \ref{prop:2players}. For instance, $(E1)$ can be written $ c(r_2-r_2)^2 + (r_1+r_2) \geq 1+\frac{1}{4c}$, it states that it should not be the case that $r_1$ and $r_2$ are both close to each other and small: indeed, it this case the left player has a profitable deviation by deviating to the right of his opponent. 

\item In the two white areas, the game does not admit any equilibrium. The reason is a repetition of the arguments above: the references are close to each other but far from $\frac{1}{2}$ or because the references are very asymmetric (both close to $0$ or to $1$) so that a player has a profitable deviation over his opponent's location. Note that this situation never occurs when references are exactly symmetric $(r_1=1-r_2)$.
\end{enumerate}
The above figure is symmetric with respect to the linear transformation $(x',y')=(y,x)$ as players are anonymous and $(x',y')=(1-y,1-x)$ as the unit interval with uniform density is symmetric.\\
~~\\
\underline{Comments on the two illustrative examples mentioned in subsection \ref{subse:2examples}.}\\ The equilibrium $(\frac{1}{2},\frac{1}{2})$ is the unique equilibrium in the game without references. In the field of political economics, it illustrates the \textit{median voter theorem}: at equilibrium both political parties select the median voter's opinion (here $\frac{1}{2}$ as voters are uniformly distributed). 
When applied to this example, Proposition \ref{prop:2players} explicitly narrows the domain of validity of the median voter theorem. It shows how political parties might differentiate when they have to compromise between \textit{office motivations} and \textit{policy interests}.\\
~~\\
In the second motivating example described in subsection \ref{subse:2examples}, the principle of minimal differentiation implies the production of two products with the same characteristic. \citep{d1979hotelling} challenges this principle by introducing a different trade off between transportation costs and price differentiation. It introduces the principle of maximal differentiation that states on the contrary that two competitors will select opposite side of the space they are competing in (in our case $0$ and $1$) in order to avoid a high price competition. In our setting the price is fixed, we explain the differentiation with only spatial arguments, the reference locations. Proposition \ref{prop:2players} shows that this ingredient is more subtle: depending on the references, we may or not observe differentiation. Note that equilibria are inefficient as the total clientele is fixed but players always pay a cost because they select $x_i \neq r_i$ at the duopoly equilibrium. However, the inefficiency does not necessarily increases with $c$, consider the case of a differentiated equilibrium where the penalty paid by each player is $c\left(\frac{1}{4c}\right)^2=\frac{1}{16c}$ which decreases with $c>0$.

\subsection{Comparative statics}\label{subse:comparative_statics}

The trade-off between strategic opportunities and costs minimization relies on the magnitude of the cost function. When $\gamma(d)=cd^2$, the parameter $c$ captures the importance of the references locations for players so we now analyze how the equilibrium configurations varies with $c \in (0,+\infty)$.


\begin{remark}\label{rmq:interpret_c} For a better interpretation of the numerical values taken by the parameter $c$, we can consider the case of a differentiated equilibrium where players
select locations at a distance $\delta=\frac{1}{4c}$ from their references. In the first motivating example, when supposing that $c=1$, we implicitly assume that the compromise found by a political party between its policy interested members and its office motivated members it that the party should deviate at a distance equal $\frac{1}{4}$ from its reference towards a strategic opportunity, so $25\%$ of the total political spectrum size. Taking $c=2.5$ leads to a compromise of a $10 \%$ deviation.
\end{remark}
The figure below describes the equilibrium configuration with respect to $(r_1,r_2) \in [0,1]^2$ for increasing values of $c$.

\begin{figure}[H] 
\begin{center}\includegraphics[scale=0.25]{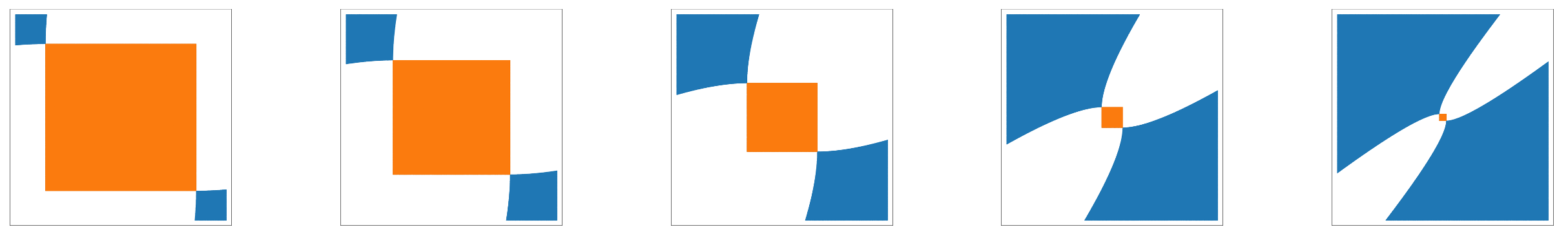}
\end{center}
\caption{The equilibrium configurations with respect to $(r_1,r_2) \in [0,1]^2$, for $c \in \{0.7,0.9,1.5,5,15\}$. In white there is no equilibrium, in orange $(\frac12,\frac12)$ is the unique equilibrium, in blue $(r_1+\frac{1}{4c},r_2-\frac{1}{4c})$ is the unique equilibrium.}
\label{fi:comparative_statics}
\end{figure}

We observe (and prove in Proposition \ref{prop:2players_otherway} below) that when $c$ goes to $0$, $(\frac12,\frac12)$ is the unique equilibrium for every profile $(r_1,r_2) \in [0,1]^2$, as it is when $c=0$. On the other hand, when $c$ goes to infinity, $(r_1+\frac{1}{4c},r_2-\frac{1}{4c})$ is the unique equilibrium for every profile $(r_1,r_2)$ such that $r_1 \neq r_2$.\\
~~\\
These two extreme cases indicates that the probability of existence of equilibrium, when the profile of references is uniformly drawn on the square $[0,1]^2$, is a non-monotonic function of $c$.\footnote{We compute the likeliness of an equilibrium in a game where each reference would be drawn at random at the beginning of the game, as it is the case for the favorite policies of candidates with character in \citep{kartik2007signaling} for example.} The probability we compute is equal to the area of the blue and orange domains together.

\begin{proposition}\label{prop:theta} Suppose that $\textbf{r}=(r_1,r_2)$ is uniformly drawn on the square $[0,1]^2$.
The game admits an equilibrium with probability 
$$\mathbb{P}(c)=\min\left(1,1-\frac{(2+4c)^{\frac{3}{2}}-6c-5}{6c^2}\right)$$
There exists $\theta \simeq 1.321$ such that $\mathbb{P}(c)$ is constantly equal to 1 for $c \leq \frac12$, decreasing for $c \in [\frac12,\theta]$ and increasing for $c \geq \theta$.
\end{proposition}

The proof of this proposition is given in \ref{proof:prop_theta}. We prove there that $\theta$ is a unique real solution of a third degree equation.\\
~~\\
\underline{Illustration:} In the following figure, the function $\mathbb{P}(c)$ is drawn in blue. Simulations using the algorithm described in subsection \ref{se:algo} confirm the theory with red dots.
\begin{figure}[H] 
\begin{center}
\includegraphics[scale=0.18]{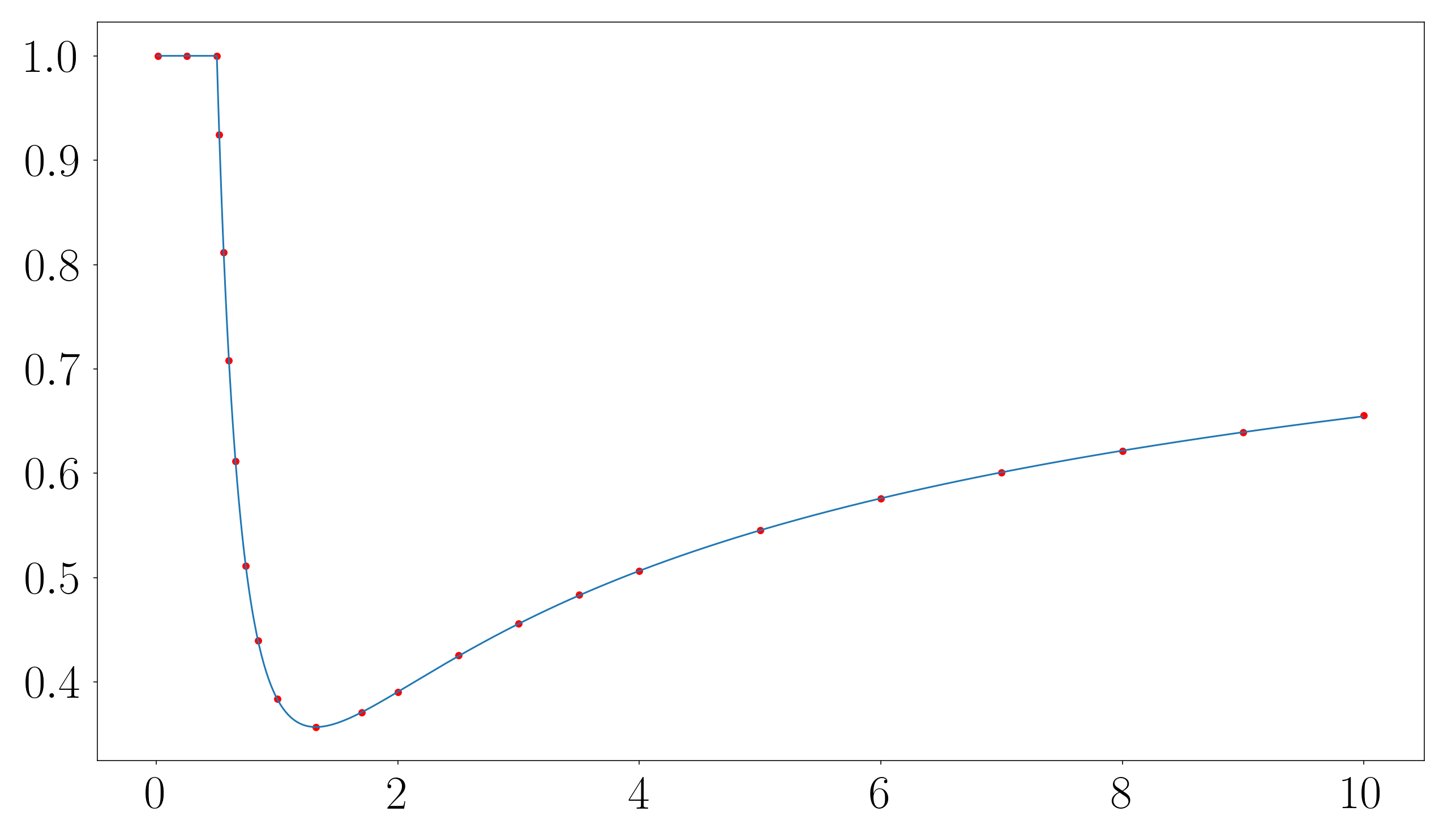}
\caption{Probability of existence of an equilibrium if $(r_1,r_2) \sim U([0,1]^2)$ with respect to $c \in [0,10]$}
\label{fi:proba_c}
\end{center}
\end{figure}
\underline{Interpretation:}
\begin{itemize}
\item For $c \leq \frac{1}{2}$, the first condition in Proposition \ref{prop:2players} is always satisfied: $$c \leq \frac{1}{2} \Rightarrow r_2-\frac{1}{4c}\leq r_1+\frac{1}{4c} \text{ for every } r_1 \leq r_2 \in [0,1]^2$$ Therefore, $(\frac12,\frac12)$ is the unique equilibrium for any references $(r_1,r_2)$. The equilibrium is the same as in the case $c=0$ because the cost function is too small to have any impact: the marginal cost of deviating is always smaller than $\frac12$, which is the marginal gain of a strategic opportunity.
\item For $c \in [\frac{1}{2},\theta]$, the probability that an equilibrium exists decreases with $c$. In this interval, the likeliness of the undifferentiated equilibrium decreases, the likeliness of the differentiated equilibrium increases and the first trend is stronger than the second one. We find that $\theta \simeq 1.321$: according to Remark \ref{rmq:interpret_c}, this value is relatively small, as it represents players that would opportunistically agree to deviate by a distance of $\frac{1}{4\theta}\simeq 0.189$ from their references, in the unit interval.
 
\item For $c \geq \theta$, the probability that an equilibrium exists increases with $c$. In this interval, the undifferentiated equilibrium has almost disappeared, so the increase of the likeliness of the differentiated equilibrium makes the probability of existence larger.
\end{itemize}

\begin{proposition}
Suppose that $\textbf{r}=(r_1,r_2)$ is uniformly drawn in the unit square. The differentiated equilibrium $(r_1+\delta,r_2-\delta)$ is more likely than the undifferentiated equilibrium $(\frac{1}{2},\frac{1}{2})$ if and only if $c \geq \eta \simeq 1.17692$.
\end{proposition}

The proof of this results is somehow trivial as Proposition \ref{prop:2players} immediately implies that the probability that the unique equilibrium is $(\frac{1}{2},\frac{1}{2})$ is $(\frac{1}{2c})^2=\frac{1}{4c^2}$. On the other hand we easily compute the probability of the differentiated equilibrium: $\mathbb{P}(c)-\frac{1}{4c^2}$ and the proposition follows. To obtain the algebraic form of $\eta$, one needs to find the positive solution to the equation $36c^4+8c^3-36c^2-24c-4=0$. Here we only provide the numerical approximation. Again, according to Remark \ref{rmq:interpret_c}, the numerical value of $\eta$ is relatively small, as it represents players that would opportunistically agree to deviate by a distance of $\frac{1}{4\eta}\simeq 0.212$ from their references, in the unit interval.\\
~~\\
In the current paper, we choose to analyze the equilibrium configuration as a function of the pair $(r_1,r_2)$ and for a fixed parameter $c$. The following proposition rephrases Proposition \ref{prop:2players} when the pair $(r_1,r_2)$ is fixed but the parameter $c$ varies.

\begin{proposition}\label{prop:2players_otherway}
Let $(r_1,r_2) \in [0,1]^2$. There exists $c_0>0$ such that for every $c \in [0,c_0]$ the game admits $(\frac12,\frac12)$ as the unique equilibrium of the game.\\
If $r_1 \neq r_2$, there also exists $c_1$ such that for every $c \in [c_1,+\infty[$ the game admits $(r_1+\delta,r_2-\delta)$ as the unique equilibrium of the game.
\end{proposition}

The previous proposition is illustrated in figure \ref{fi:comparative_statics} where we can see that for any pair of references, the equilibrium is undifferentiated for $c$ small enough, and differentiated for $c$ large enough (and $r_1 \neq r_2$). The proof is straightforward: for $c$ small enough it is always the case that $r_2-\frac{1}{4c} \leq \frac12 \leq r_1+\frac{1}{4c}$ so Proposition \ref{prof:prop2players} states that there exists a unique equilibrium which is $(\frac12,\frac12)$. On the other hand, if $c$ is large enough and $r_1 \neq r_2$, take without loss of generality $r_1<r_2$, then is it always the case that $r_1+\frac{1}{4c}<r_2-\frac{1}{4c}$ and both inequalities $(E1)$ and $(E2)$ hold, so Proposition \ref{prof:prop2players} provides the existence and uniqueness of the differentiated equilibrium.\\
~~\\
This last proposition opens a mechanism design perspective: if the planner can decide which parameter $c>0$ to implement, for example in the taxation setting mentioned in the introduction, she basically decides whether there is an equilibrium or not, and whether the equilibrium is differentiated or not. Moreover, because the distance between players at a differentiated equilibrium is given by $d(x_2,x_1)=r_2-r_1-\frac{1}{2c}$, she can also decide the distance between firms in the interval $[0,r_2-r_1)$.

\subsection{General cost function $\gamma$}\label{subse:general_costs}

Proposition \ref{prop:2players} can be generalized to the case where the cost function is given by a strictly increasing and strictly convex function $\gamma: [0,1] \rightarrow \mathbb{R}$. In this case, the following parameter $\delta$ plays an important role.

\begin{definition}\label{def:delta} We denote $\delta$ the unique distance in $[0,1]$, if it exists, such that $\gamma'(\delta)=\frac{1}{2}$. If $\gamma'(d) < \frac{1}{2}$ for every $d \in [0,1]$ we set $\delta:= +\infty$ and if $\gamma'(d) > \frac{1}{2}$ for every $d \in [0,1]$ we set $\delta:= 0$.
\end{definition}

The uniqueness of $\delta$ is a consequence of the strict convexity of $\gamma$. As noted in the discussion of proposition \ref{prop:2players}, the marginal gain when moving closer to an opponent is constant equal to $\frac12$, while the marginal cost of a deviation further the reference is increasing as the cost function $\gamma$ is convex. The parameter $\delta$ captures the turning point where the second becomes larger than the first, therefore $\delta$ is the maximal distance that a player beneficially deviates from its reference in the same neighborhood. Note that in the particular case where $\gamma(d)=cd^2$ then $\delta=\frac{1}{4c}$.

\begin{proposition}\label{prop:2p_generalcost}
Let $\textbf{r}=(r_1,r_2) \in [0,1]^2$. Suppose, without loss of generality, that $r_1 \leq r_2$. Then we have:
\begin{itemize}
\item If $r_2-\delta \leq r_1+\delta$, the unique possible equilibrium is $(\frac{1}{2},\frac{1}{2})$.\\~~\\ It is an equilibrium if and only if $\frac12 \in [r_2-\delta,r_1+\delta]$.
\item If $r_1+\delta<r_2-\delta$, the unique possible equilibrium is $(r_1+\delta,r_2-\delta)$.\\~~\\ It is an equilibrium if and only if $\begin{cases} \gamma(r_2-r_1-\delta)-\gamma(\delta)+\frac{r_1+3r_2}{2}-1-\delta \geq 0~~~~(E1')\\ \gamma(r_2-r_1-\delta)-\gamma(\delta)-\frac{3r_1+r_2}{2}+1-\delta \geq 0~~~~(E2') \end{cases}$
\end{itemize}
\end{proposition}

The proof of Proposition \ref{prop:2p_generalcost} strictly follows the proof of Proposition \ref{prop:2players} and the interpretation is unchanged. 
Remark that, although the unique equilibrium locations only depend on $\delta$ (and not on the specific function $\gamma$), the conditions for this candidate to be indeed an equilibrium depends on $\gamma$, as we can see in the inequalities $(E1')$ and $(E2')$. 

\subsection{Heterogeneous costs}\label{subse:hetero_costs}

We show in this subsection that Proposition \ref{prop:2players} can be easily extended to the case where players incur  different cost functions $\gamma_i(d) = c_id^{2}$. We denote $\delta_i=\frac{1}{4c_i}$. 

\begin{proposition}\label{prop:2players_hetero}
Let $r=(r_1,r_2)\in[0,1]^{2}$. Suppose, without loss of generality, that $r_1 \leq r_2$. Then we have :
\begin{itemize}
\item If $r_2-\delta_2 \leq r_1+\delta_1$, the unique possible equilibrium is $(\frac{1}{2},\frac{1}{2})$.\\ It is an equilibrium if and only if $\frac{1}{2}\in[r_2-\delta_2,r_2+\delta_2] \cap [r_1-\delta_1,r_1+\delta_1]$.
\item If $r_1+\delta_1<r_2-\delta_2$, the unique possible equilibrium is $(r_1+\delta_1,r_2-\delta_2)$.\\ It is an equilibrium if and only if:\\ $\begin{cases}c_1(r_2-r_1)^2 - 2c_1\delta_2(r_2-r_1)+\frac{1}{2}(r_1+\delta_1)+\frac{3}{2}(r_2-\delta_2)+c_1(\delta_2^{2}-\delta_1^{2})-1 \geq 0 ~~ (E1'')\\ c_2(r_2-r_1)^2 - 2c_2\delta_1(r_2-r_1)-\frac{1}{2}(r_2-\delta_2)-\frac{3}{2}(r_1+\delta_1)-c_2(\delta_2^{2}-\delta_1^{2})+1 \geq 0 ~~ (E2'')\end{cases}$
\end{itemize}
\end{proposition}

The proof is postponed to subsection \ref{proof:prop2players_hetero}.\\
~~\\
\underline{Illustration:} the above figures are obtained 
with $c_1=1$ and $c_2=2$ on the left, and with $c_1=1$ and $c_2=10$ on the right.

\begin{figure}[H]
\begin{center}
\includegraphics[scale=0.25]{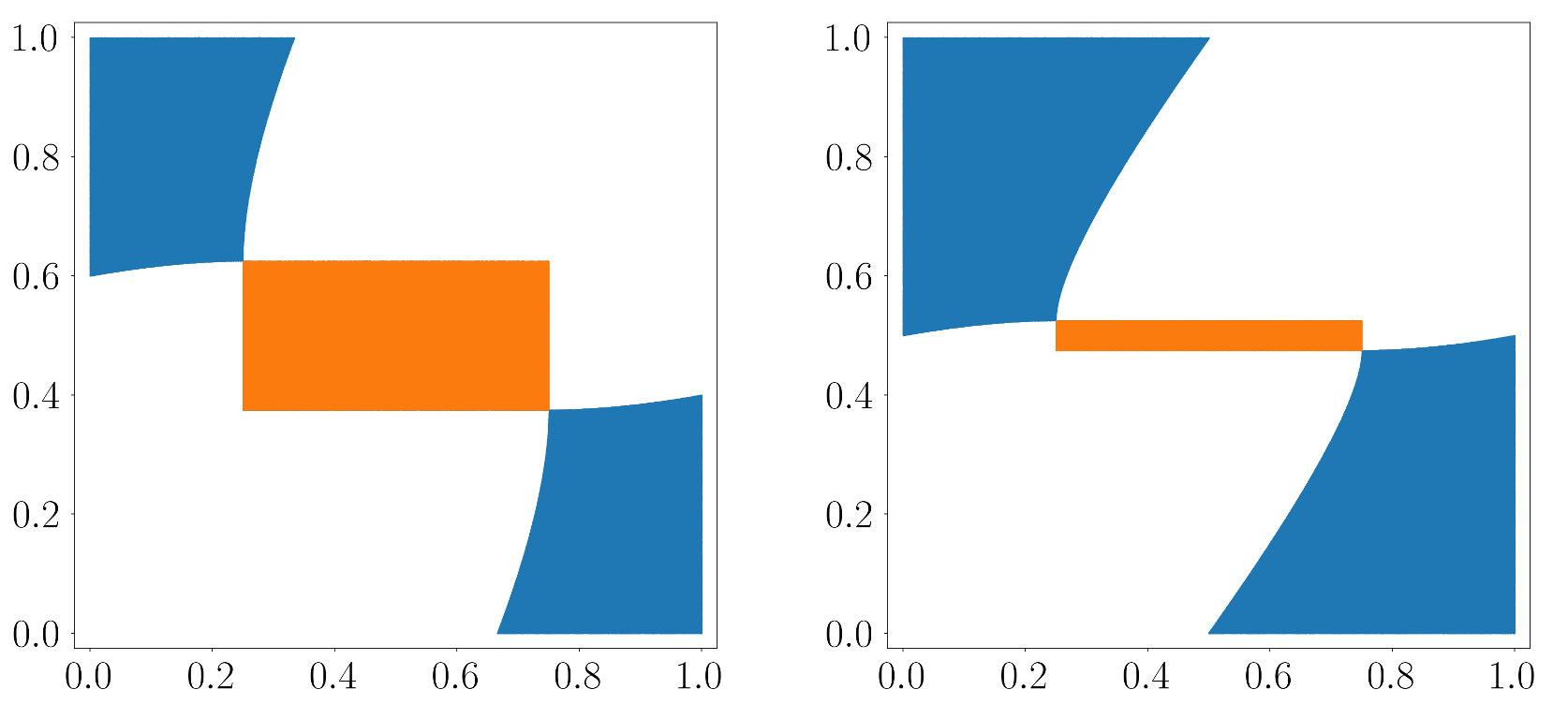}
\end{center}
\caption{The equilibrium structures with asymmetric costs: $(c_1,c_2)=(1,2)$ on the left, and $(1,10)$ on the right. In white there is no equilibrium, in orange $(\frac12,\frac12)$ is the unique equilibrium, in blue $(r_1+\delta_1,r_2-\delta_2)$ is the unique equilibrium.}
\label{fi:hetero}
\end{figure}
\underline{Interpretation:} The structure of equilibria remains unchanged: when references are far from each other there is a differentiated equilibrium, if they are close to each other and close to $\frac{1}{2}$, the unique equilibrium is $(\frac12,\frac12)$. In this second case, the equilibrium conditions are modified and do not define a square anymore but a rectangle. Finally if they are close to each other but far from $\frac12$ there is no equilibrium. \\
The figures are still invariant to the linear transformation $(x',y')=(1-x,1-y)$ as the unit interval with uniform density is symmetric, but the figures are not symmetric with respect to the transformations $(x',y')=(y,x)$ or $(x',y')=(1-y,1-x)$ anymore, as it was the case when costs where homogeneous, because players are not anonymous anymore.

\section{The triopoly}\label{se:triopoly}

\subsection{Existence and characterization of equilibrium}\label{subse:equilibrium_3p}

The case where three players compete on the unit interval is particularly interesting as it is the unique case where there is no pure equilibrium in the game without cost (where $\gamma=0$). Indeed, it is easy to prove that at equilibrium the two following contradictory statements hold:\\ (1) Peripheral players are paired,\\ 
(2) No more than two players share the same location.\\
An immediate consequence is that the triopoly does not admit any equilibrium.\\ 

Different papers 
tried to fill the gap, by introducing mixed equilibrium (\cite{shaked1982existence}) or by relaxing the assumption that customers always buy to the closest shop (\cite{de1987existence}) for example.\\

In this section we show that the introduction of reference locations has an impact on the existence of a triopoly equilibrium. More precisely: the non-existence result is robust to the introduction of small reference costs. When the costs increase above a certain threshold, there exists either a unique equilibrium or no equilibrium, depending on the references. When the cost is asymptotically large, there exists an equilibrium for almost every triplet of references.

\begin{proposition}\label{prop:58noNE}
Suppose that $\gamma(d)=cd^2$ and consider the triopoly competition.\\
(1) If $c < \frac{5}{8}$, there is no equilibrium for any triplet of references $\textbf{r}$.\\
(2) If $c \geq \frac{5}{8}$, there exists either a unique or no equilibrium depending on the triplet of references.\\ 
(3) For every triplet of different references $r_1<r_2<r_3$, there exists a unique equilibrium for $c$ large enough.
\end{proposition}

This proposition is proved in \ref{proof:58noNE}. In the proof we first characterize the unique equilibrium candidate $\textbf{x}^{\star}$ in \ref{proof:candidate_triopoly} that we now describe, after introducing the useful notion of far-left and far-right players. 

\begin{definition}\label{def:periph_far_players} 
If $r_1+\frac{1}{4c}<r_2$ we say that Player $1$ is a far-left player. If $r_n-\frac{1}{4c}>r_{n-1}$ we say that Player $n$ is a far-right player.\\
In the general case of strictly convex and increasing cost function $\gamma$, we say that Player $1$ is  far-left if $r_1+\delta<r_2$ and that Player $n$ is far right if $r_n-\delta>r_{n-1}$, where $\delta$ is defined in \ref{def:delta}.
\end{definition}

We emphasize that the notion of far player concerns a profile of references $\textbf{r}$ while the notion of peripheral player is a property of a profile of locations $\textbf{x}$. 
As shown in Proposition \ref{prop:58noNE} for three players, and later in Proposition \ref{thm:equilibrium_description} for the general case, a far-left or far-right player has a reference which is so extreme that he never shares his equilibrium location with another player. 
We now describe the equilibrium candidate in the triopoly competition (proof in \ref{proof:candidate_triopoly}). Suppose without loss of generality that $r_1 \leq r_2 \leq r_3$, then:
\begin{itemize}
\item If Player $1$ is far-left and Player $3$ is far-right then the unique possible equilibrium is $\textbf{x}^{\star}=(r_1+\frac{1}{4c},r_2,r_3-\frac{1}{4c})$. 
\item If Player $1$ is not far-left but Player $3$ is far-right, the unique possible equilibrium is $\textbf{x}^{\star}=\left(\frac{r_3-1/4c}{3},\frac{r_3-1/4c}{3},r_3-\frac{1}{4c}\right)$, and symmetrically, if Player $1$ is far-left but Player $3$ is not far-right, the unique possible equilibrium is $\textbf{x}^{\star}=\left(r_1+\frac{1}{4c},\frac{r_1+1/4c+2}{3},\frac{r_1+1/4c+2}{3}\right)$.
\item If neither Player $1$ nor Player $3$ are far-left or far-right, there is no equilibrium.
\end{itemize}

The unique equilibrium candidate $\textbf{x}^{\star}$ is found by considering the necessary condition that at equilibrium players do not have profitable \textit{local deviation}, that is a deviation within the same neighborhood. To fully characterize the set of equilibria, one also has to consider deviations to a different neighborhood. We postpone this discussion to section \ref{se:general_case}: we show that, in the general case of $n$ players, to insure that no players has a deviation in a different neighborhood, the vector of reference locations has to satisfy a system of $n(n-1)$ equations. We now provide intuition and illustration for Proposition \ref{prop:58noNE}.\\


The intuitive reason why the classical non-existence result in the triopoly competition is robust against the introduction of a small cost of deviation from a reference is the following:\\
- If the references are close to $\frac12$, the argument is the same as in the case without costs: all three players want to locate together at the center, but such a configuration is not stable as each player could deviate marginally to his left or right and get half of the interval instead of sharing a third.\\
- Therefore, the most favorable case for existence of an equilibrium is the case where $(r_1,r_2,r_3)=(0,\frac12,1)$. The description of the unique equilibrium candidate above gives that $x_1=\frac{1}{4c}$, $x_2=\frac12$ and $x_3=1-\frac{1}{4c}$, and we find that if $c\leq \frac58$, Player $2$'s payoff is $\frac{1}{5}$ which makes a deviation to the left of Player $1$ (or to the right of Player $2$) profitable.\\


\underline{Illustration:} The graph below is concerned with the existence of equilibrium in the triopoly competition.

\begin{figure}[H]
\begin{center}
\includegraphics[scale=0.17]{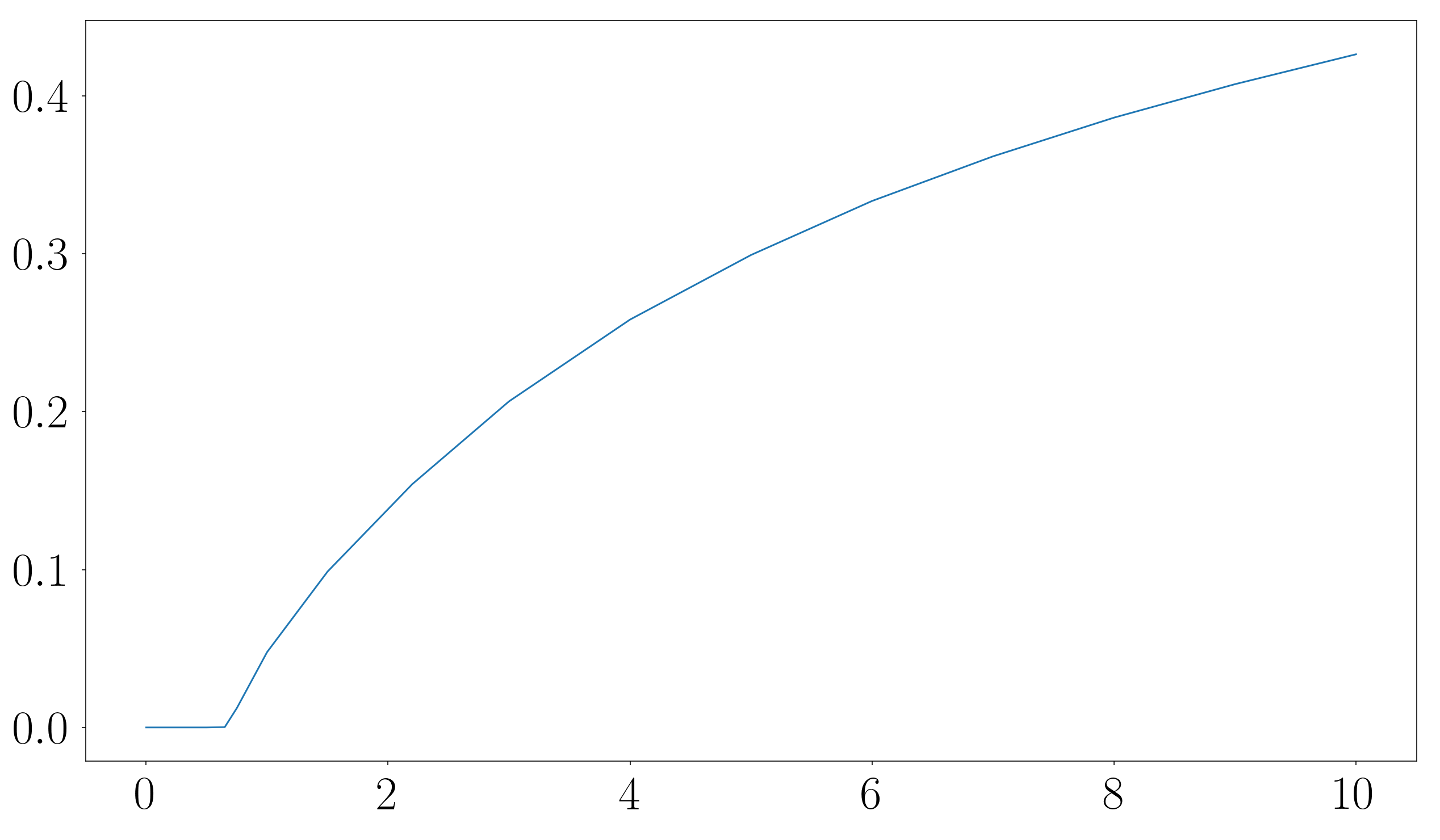}
\caption{Probability that the triopoly competition admits an equilibrium when $\textbf{r}\sim \mathcal{U}([0,1]^3)$ and $\gamma(d)=cd^2$ for $c \in (0,10)$.}
\label{fi:proba_c_triopoly}
\end{center}
\end{figure}

The graph illustrates the fact that for $c$ small enough, the probability that an equilibrium exists is zero (Proposition \ref{prop:58noNE} is stronger as it also shows that there does not exist a negligible set of equilibria neither). Although the probability of existence goes to $1$ as $c \rightarrow +\infty$, it increases slowly: half of the profiles provide an equilibrium only when $c \simeq 15.35$. According to Remark \ref{rmq:interpret_c}, this numerical value of $c$ is relatively large, as it represents players that would opportunistically agree to deviate only by a distance of $\frac{1}{4c}\simeq 0.016$ from their references, in the unit interval.\\
~~\\
As an illustration, we now analyze the triopoly competition in two particular cases.

\subsection{Illustrative example 1: symmetric references profiles}\label{subse:symmetric3p}

We consider the particular case where the references profiles are symmetric. Therefore, we suppose that $(r_1,r_2,r_3)=(r_1,\frac{1}{2},1-r_1)$ with $r_1 \in [0,\frac{1}{2})$.

\begin{proposition}\label{prop:triopoly_symmetric}
Suppose that $\gamma(d)=cd^2$. The triopoly competition with references $\textbf{r}=(r_1,\frac{1}{2},1-r_1)$ admits a unique possible equilibrium which is $\textbf{x}^{\star}=(r_1+\frac{1}{4c},\frac{1}{2},1-r_1-\frac{1}{4c})$. This is an equilibrium if and only if $r_1 \leq \phi(c):=\frac{3+2c-2\sqrt{4+2c}}{4c}$.
\end{proposition}
The proof is postponed to subsection \ref{proof:triopoly_symmetric}.\\
~~\\
This result illustrates Proposition \ref{prop:58noNE} as $\phi$ is an increasing function: if $c$ is small enough, there exists no equilibrium. We can compute that for $c \leq \frac{2\sqrt{2}-1}{2} \simeq 0.914$ there exists no equilibrium for any symmetric profile of location. It improves the general threshold of $\frac58$ provided for the general case. We also observe that $\phi(c)$ goes to $\frac{1}{2}$ when $c$ goes to $+\infty$, so that the game admits a unique equilibrium for every $r_1 \in [0,\frac{1}{2})$ if $c$ is large enough.

\subsection{Illustrative example 2: one reference is fixed to $\frac12$}\label{subse:algo3p}

The following figure illustrates the equilibrium characterized in subsection \ref{subse:equilibrium_3p} in the case where a reference is fixed to $\frac{1}{2}$ and $\gamma(d)=10d^{2}$. 

\begin{figure}[H]
\begin{center}
\includegraphics[scale=0.15]{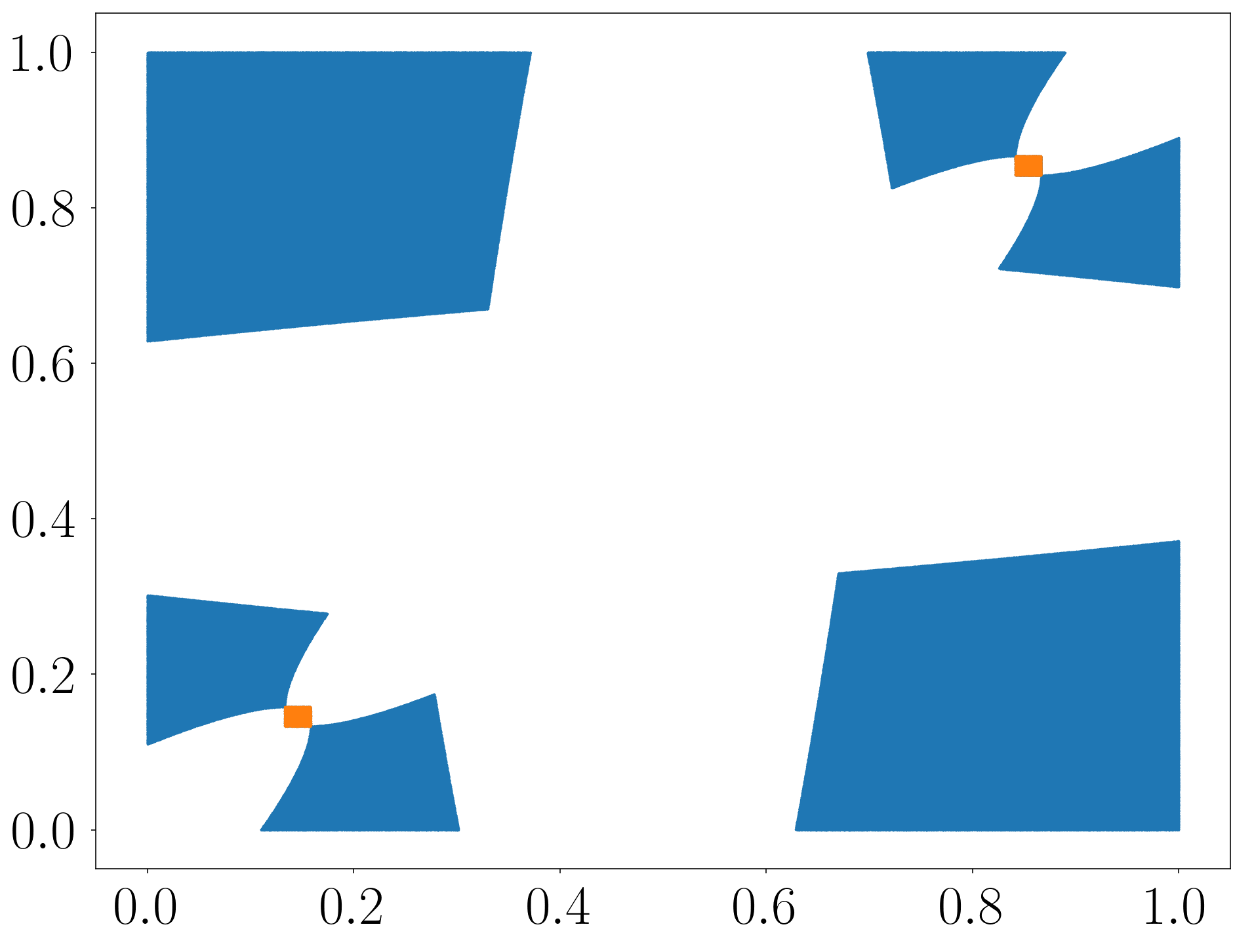}
\end{center}
\caption{The equilibrium structure with respect to $(r_1,r_2)\in[0,1]^2$ when $r_3=\frac{1}{2}$. In blue: there exists a unique equilibrium where all players select different locations, in orange: there exists a unique equilibrium where two players select the same location, in white there is no equilibrium.}
\label{fi:triopoly_one_is_fixed}
\end{figure}


When both $r_1$ and $r_2$ are far from $\frac12$, the figure show 4 symmetrical sub-figures:\\ 
- the bottom-left sub-figure illustrates the case where $r_1 \leq r_2 \leq r_3=\frac{1}{2}$. This sub-figure is quite similar to the figure \ref{fi:duopoly} obtained in the duopoly competition. This is explained by the fact that the third player locates at $r_3=\frac12-\frac{1}{4c}$ and that the two players compete on the sub-interval $[0,x_3)$. A deviation to the right of $x_3$ would be too costly. As in the duopoly competition, the orange rectangle represents the cases where the two players select the same location, which is here $\frac16-\frac{1}{12c}$, i.e. the location where they attract the same quantity of resources from the left and from the right when the third player is located at $\frac12-\frac{1}{4c}$ (because claim 2 of Lemma \ref{prop:3claims} applies).\\
- take now the top left sub-figure, concerned with the case where $r_1 \leq r_3=\frac{1}{2} \leq r_2$. In this case, we find that the equilibrium is differentiated, $\textbf{x}=(r_1+\frac{1}{4c},\frac{1}{2},r_2-\frac{1}{4c})$ and the condition for this profile to be an equilibrium is that there is no profitable deviation of Player $2$ or $3$ to the left of Player $1$, and no profitable deviation of Player $1$ or $3$ to the left of Player $2$. These conditions involve quadratic equation that are similar to $(E1)$ and $(E2)$ in Proposition \ref{prof:prop2players}.\\
The top-right and bottom-right sub-figures are symmetric to the previous ones. Indeed, the figure is symmetric with respect to the linear transformation $(x',y')=(y,x)$ as players are anonymous and $(x',y')=(1-y,1-x)$ as the unit interval with uniform density is symmetric.

\section{The general case of $n$ players}\label{se:general_case}

In this section, we only suppose that the cost function $\gamma$ is a $C^1$, strictly convex and strictly increasing function. We first prove that for any vector of references the game admits at most one equilibrium and we characterize the unique equilibrium candidate and its properties. We then consider the sufficient equilibrium conditions and we prove that these (infinitely many) conditions can be reduced to a list of $n(n-1)$ equations. Finally, we illustrate our result with the particular case of regularly located references.

\subsection{Existence, uniqueness and characterization of the equilibrium}\label{subse:existence}
 
The next proposition deals with the case of $n \geq 5$ players. The duopoly and triopoly are studied in section \ref{se:duopoly} and \ref{se:triopoly}. The case where $n=4$ is also particular because it is possible that all players are peripheral, which is not possible when $n \geq 5$ as proved in Claim 1 of Proposition \ref{prop:3claims}. For the sake of completeness, the case of $4$ player is provided in subsection \ref{subsec:4players} and poses no conceptual difficulties.\\
~~\\ 
Proposition \ref{thm:equilibrium_description} exhibits the unique equilibrium candidate and Proposition \ref{prop:sufficient_general} provides necessary and sufficient conditions for this candidate to be an equilibrium. The parameter $\delta$ is defined in \ref{def:delta} and the notions of far-left and far-right players are introduced in \ref{def:periph_far_players}. 

\begin{proposition}\label{thm:equilibrium_description}
Let $n \geq 5$ and $r_1 \leq r_2 \leq \dots \leq r_n$. The unique possible equilibrium $\textbf{x}^{\star}$ of the game with references $\textbf{r}$ is the following:
\begin{itemize}
\item If Player $1$ is a far-left player, he is the only left-peripheral player and locates at $x_1^{\star}=r_1+\delta$. Otherwise, both Player $1$ and Player $2$ are the left-peripheral players and they locate at $\frac{r_3}{3}$.
\item Symmetrically, if Player $n$ is a far-right player, he is the only right-peripheral player and locates at $x_n^{\star}=r_n-\delta$. Otherwise, both Player $n$ and Player $n-1$ are the right-peripheral players and they locate at $\frac{2+r_{n-2}}{3}$.
\item All non peripheral players locate at their references: $x_i^{\star}=r_i$.
\end{itemize}
\end{proposition}

The proof is given in subsection \ref{proof:general}. We now discuss the interesting properties of the equilibrium candidate:\\
(1) the equilibrium, when it exists, is unique. This property differs from the analysis of location games without references, where there exists for $n>6$ a continuum of equilibria (a polyhedron of dimension $n-5$). The reason for uniqueness in the current paper is that the cost of deviating from the reference location is a tie-breaking rule: while the quantity of resources that a player attracts does not depend on the exact position of a player in a fixed neighborhood, the total payoff is not constant as the costs are strictly increasing with the distance.\\
(2) at equilibrium, at most $4$ players select a location which is different from their references: the two most left and most right players. Therefore, when studying the game with many players, only the border effects have to be looked at, i.e. the behavior of players $1$ and $2$ (symmetrically of players $n-1$ and $n$). All the interior players locate at their references.\\
(3) the border effects are the same for any $n\geq 3$: either player $1$'s reference is far enough to player $2$'s reference and then player $1$ acts as a monopolist and satisfies his first order conditions, or the $2$ most left references are close enough and players select the same location. In this case the choice of this location is dictated by \ref{prop:3claims}, claim 2: they must be located at $\frac{1}{3}$ of the interval $[0,r_3]$. The argument is symmetric on the right.\\
~~\\
We illustrate Proposition \ref{thm:equilibrium_description} with a numerical example.

\begin{example}\label{ex:illustration_general_thm} Suppose that $\bold{r}=(0.10,0.13,0.42,0.52,0.75,0.95)\in[0,1]^{6}$ and $\gamma(d)=5d^2$. Player 1 is not far-left since $r_2-r_1 \leq \delta = 0.05$, but Player $6$ is far-right since $r_6-r_5>\delta$. The profile $\textbf{x}^{\star}$ is therefore such that $x_1=x_2=r_3/3$ and such that $x_6=r_6-\delta$. All other players selects $x_i=r_i$.
Thus, $\bold{x}^{\star}=(\frac{r_3}{3},\frac{r_3}{3},r_3,r_4,r_5,r_6-\delta)=(0.14,0.14,0.42,0.52,0.75,0.90)$ is the unique equilibrium candidate.

\begin{figure}[H]\label{figure:example}
\centering
\begin{tikzpicture}[>=latex]

    \draw[-] (0,0) --(10,0);
    \node[] at (0.9,0) {$|$};
    \node[] at (1.2,0) {$|$};
    \node[] at (4.2,0) {$|$};
    \node[] at (5.2,0) {$|$};
    \node[] at (7.5,0) {$|$};
    \node[] at (9.5,0) {$|$};

    \draw[draw=red,line width=1pt] (0,0)-- (2.9,0);
    \draw[draw=green,line width=1pt] (2.9,0)-- (4.7,0);
    \draw[draw=blue,line width=1pt] (4.7,0)-- (6.35,0);
    \draw[draw=yellow,line width=1pt] (6.35,0)-- (8.25,0);
    \draw[draw=orange,line width=1pt] (8.25,0)-- (9.98,0);

    \node[below=8pt] at (0,0) {$0$};
    \node[below=8pt] at (0.85,0) {$r_1$};
    \node[below=8pt] at (1.25,0) {$r_2$};
    \node[below=8pt] at (4.2,0) {$r_3$};
    \node[below=8pt] at (5.2,0) {$r_4$};
    \node[below=8pt] at (7.5,0) {$r_5$};
    \node[below=8pt] at (9.5,0) {$r_6$};
    \node[below=8pt] at (10,0) {$1$};

    \tikzset{decorate sep/.style 2 args=
{decorate,decoration={shape backgrounds,shape=circle,shape size=#1,shape sep=#2}}}

    \node[circle,fill,inner sep=1.9pt] at (1.5,0) {};
    \node[circle,fill,inner sep=1.9pt] at (4.2,0) {};
    \node[circle,fill,inner sep=1.9pt] at (5.2,0) {};
    \node[circle,fill,inner sep=1.9pt] at (7.5,0) {};
    \node[circle,fill,inner sep=1.9pt] at (9,0) {};

    \node[above=8pt] at (1.5,0) {$x_1$};
    \node[above=16pt] at (1.5,0) {$x_2$};
    \node[above=8pt] at (4.2,0) {$x_3$};
    \node[above=8pt] at (5.2,0) {$x_4$};
    \node[above=8pt] at (7.5,0) {$x_5$};
    \node[above=8pt] at (9,0) {$x_6$};

\end{tikzpicture}
\end{figure}
\end{example}


\begin{proposition}\label{prop:sufficient_general} 
The equilibrium candidate described in Proposition \ref{thm:equilibrium_description} is an equilibrium if and only if the three following conditions hold:\\
(1) If Player 1 is not far-left then $\frac{r_3}{3} \in [r_2,r_1+\delta]$.\\
(2) Symmetrically, if Player $n$ is not far-right then $\frac{2+r_{n-2}}{3} \in [r_n-\delta,r_{n-1}]$.\\
(3) $\Delta_i^j(\textbf{x}^{\star}) \geq 0$ for every  $i \neq j \in \{1,\dots,n\}$,
where $\Delta_i^j(\textbf{x})$ denote the loss of player $i$ when he deviates infinitesimally close to Player $j$'s location when $\textbf{x}$ is played. More precisely:
$$
\Delta_i^j(\textbf{x})= \begin{cases} 
\lim\limits_{\substack{\epsilon \to 0 \\ \epsilon>0}} u_i(\textbf{x})-u_i(x_j+\epsilon,\textbf{x}_{-i})  \text{ if } x_i<x_j\\
\lim\limits_{\substack{\epsilon \to 0 \\ \epsilon>0}} u_i(\textbf{x})-u_i(x_j-\epsilon,\textbf{x}_{-i})   \text{ if } x_i>x_j\\
0  \text{ if } x_i=x_j
\end{cases}$$
\end{proposition}

Conditions $(1)$ and $(2)$ are concerned with local deviations of peripheral players when they are paired. If player $1$ is not a far-left player, then the unique possible equilibrium described in Proposition \ref{thm:equilibrium_description} is such that $x_1=x_2=\frac{r_3}{3}$. In the proof, we show that if $\frac{r_3}{3}>r_1+\delta$ then Player $1$ has a profitable deviation: $\frac{r_3}{3}$ is too far from his reference (at a distance larger than $\delta$, which the maximal distance he can deviate from $r_1$ as a peripheral player). On the other hand, if $\frac{r_3}{3}<r_2$ then Player $2$ would improve his payoff by playing $\frac{r_3}{3}+\epsilon$ for $\epsilon>0$ small enough: the quantity of attracted resources is the same but the reference cost is smaller.\\
Condition $(3)$ is concerned with global deviation: if $\Delta_i^j(\textbf{\textit{x}}^{\star})<0$ then player $i$ has a profitable deviation by selecting a location infinitesimally close to player $j$: the market opportunity is larger than the reference cost of such a deviation. Among the infinity of possible deviations, we show that if these $n-1$ deviations are not profitable, then the profile is an equilibrium.

\begin{example} We come back to example \ref{ex:illustration_general_thm}. The equilibrium candidates is $\bold{x}^{\star}=(\frac{r_3}{3},\frac{r_3}{3},r_3,r_4,r_5,r_6-\delta)=(0.14,0.14,0.42,0.52,0.75,0.90)$.\\
It satisfies condition $(1)$ as $\frac{r_3}{3}=0.14 \in [0.13,0.15]=[r_2,r_1+\delta]$. Condition $(2)$ is empty as player $6$ is far-right. Condition $(3)$ translates to:\\
$(\Delta_1^j)_{j \neq 1}= (0,~0.594,~0.899,~2.1695,~3.232)$\\
$(\Delta_2^j)_{j \neq 2}= (0,~0.51,~0.785,~1.9865,~3.004)$\\
$(\Delta_3^j)_{j \neq 3}= (0.442,~0.442,~0.3395,~0.6595,~1.242)$\\
$(\Delta_4^j)_{j \neq 4}= (0.747,~0.747,~0.075,~0.3545,~0.787)$\\
$(\Delta_5^j)_{j \neq 5}= (1.9105,~1.9105,~0.5945,~0.4045,~0.2025)$\\
$(\Delta_6^j)_{j \neq 6}= (3.303,~3.303,~1.427,~1.037,~0.2475)$\\
we see that $\Delta_i^j \geq 0$ for every $i \neq j$ so the sufficient conditions are satisfied and $\bold{x}^{\star}$ is an equilibrium.
\end{example}

\subsection{Illustrative example: regular location of references}

In this subsection, we consider the case where references are regularly distributed on $[0,1]$ to illustrate our results.

\begin{proposition}\label{prop:regular_n_players}
In the game with $n \geq 5$ players where $\gamma(d)=cd^2$ and $r=\left(\frac{1}{n+1},\frac{2}{n+1},...,\frac{n}{n+1}\right)$, we have that:
\begin{center}
$\textbf{x}^{\star}=\left(\frac{1}{n+1}+\frac{1}{4c},\frac{2}{n+1},...,\frac{n}{n+1}-\frac{1}{4c}\right)$ is an equilibrium if and only if $c \geq \frac{1+\sqrt{6}}{4}(n+1)$.
\end{center}
Otherwise there is no equilibrium.
\end{proposition}

The proposition is proved in \ref{proof:uniform_general}. The unique equilibrium candidate is such that only players $1$ and $n$ are peripheral. The intuition for the result is the following: when $\textbf{x}^{\star}$ is played, the peripheral players are players who attract the more resources while Players $2$ and $n-1$ are players who attract the less resources. For this profile to be an equilibrium, we need $c$ to be large enough so that Players $2$ does not have a profitable deviation to deviate to the left of Player $1$. As $n$ becomes large, the distance between players decreases and the constraint on $c$ is stronger.

\subsection{The case where $n=4$}\label{subsec:4players}
For the sake of completeness, we here give the equilibrium configuration in the case of $n=4$ players.

\begin{proposition}\label{prop:4players}
Let $\textbf{r}=(r_1,r_2,r_3,r_4) \in [0,1]^4$. Suppose, without loss of generality, that $r_1 \leq r_2 \leq r_3 \leq r_4$. Then we have:
\begin{itemize}
\item If Player $1$ is a far-left player and Player $4$ is a far-right player then the unique possible equilibrium is $\textbf{x}^{\star}=(r_1+\delta,r_2,r_3,r_4-\delta)$. 
\item If Player $1$ is not a far-left player but Player $4$ is a far-right player, the unique possible equilibrium is $\textbf{x}^{\star}=(\frac{r_3}{3},\frac{r_3}{3},r_3,r_4-\delta)$. 
\item Symmetrically, if Player $1$ is a far-left player but Player $4$ is not a far-right player, the unique possible equilibrium is $\textbf{x}^{\star}=(r_1+\delta,r_2,\frac{r_2+2}{3},\frac{r_2+2}{3})$. 
\item If neither Player $1$ nor Player $4$ are far left or far-right, there is no equilibrium.
\end{itemize}
The necessary and sufficient conditions given in Proposition \ref{prop:sufficient_general}, for the equilibrium candidate to be an equilibrium, also apply for $n=4$.
\end{proposition}

This proposition is given here for the sake of completeness. The results slightly differ as it might be the case that no player select $x_i=r_i$. However, its proof poses no conceptual difficulties and relies on a case-by-case analysis using the same arguments as the general case of $n \geq 5$ players.

\subsection{Algorithm}\label{se:algo}

Figures \ref{fi:duopoly}, \ref{fi:comparative_statics}, \ref{fi:proba_c}, \ref{fi:hetero}, \ref{fi:proba_c_triopoly}, \ref{fi:triopoly_one_is_fixed} or similar figures with different parameters can be reproduced using the algorithms described in this section and available at \href{https://github.com/FournierAMSE/Location-games-with-references}{https://github.com/FournierAMSE\\/Location-games-with-references}.

Each program is concerned with a particular case of $2$, $3$, $4$ or $n \geq 5$ players. Each program typically contains the following functions:

\begin{itemize}
\item \texttt{equilibrium(\textbf{r},c)}:  For the given reputation vector $\textbf{r}$ and $c>0$, this function verifies whether the unique equilibrium candidate characterized in Proposition~\ref{thm:equilibrium_description} (for the case $n \geq 5$) is an equilibrium. To do so, it checks whether the candidate satisfies each necessary and sufficient conditions provided in Proposition~\ref{prop:sufficient_general}.
\item \texttt{probability\_NE(nb\_draws,c)}: This function estimates the probability of existence of an equilibrium by performing several random draws of reputation profiles $\textbf{r}$ and computing the ratio of equilibrium to the number of draws (NE/nb\_draws).
\item \texttt{NE\_plot(nb\_draws,c)} : This function performs several random draws of $(r_1,r_2)$ and plots the couples for which there exists an equilibrium. When plotting these figures with $n>2$, we need to fix $n-2$ reputations to obtain a $2$-dimensional figure.
\end{itemize}

\section{Proofs}\label{se:proofs}

For simplicity, we use the vocabulary from the second motivating example presented in the introduction: resources are customers, players are sellers, the consumers buying to a given seller are its clientele. We first introduce some useful notations:

\begin{itemize}
\item For a profile $\textbf{x}=(x_1,\dots,x_n)$ of locations in $[0,1]^n$ and a point $y \in [0,1]$, we denote $\card(y,\textbf{x})$ the number of players who selected location $y$ when the profile $\textbf{x}$ is played. $\card(y,\textbf{x})=0$ when $y \notin \{ x_1,\dots,x_n\}$. If $\card(y,\textbf{x})=2$ we say that $2$ players are paired in $y$.

\item We denote by $q_{i}(\textbf{x})$ the quantity of customers attracted by Player $i$. These consumers are called his \textit{clientele}. 
With the notations introduced at the end of section \ref{se:model}, we have:
$$q_{i}(\textbf{x}):= \frac{x_{i}^r-x_{i}^{\ell}}{2 \card(x_i,\textbf{x})}$$
It will sometimes be useful to consider $q_i^{\ell}(x)$ and $q_i^{r}(x)$ the quantity of customers attracted by player $i$ from the left and from the right:
$$q_{i}^{\ell}(\textbf{x}):= \frac{x_{i}-x_{i}^{\ell}}{2 \card(x_i,\textbf{x})}$$
$$q_{i}^{r}(\textbf{x}):= \frac{x_{i}^r-x_{i}}{2 \card(x_i,\textbf{x})}$$
We obviously have the relation: $q_{i}(\textbf{x})=q_{i}^{\ell}(\textbf{x})+q_{i}^{r}(\textbf{x})$.
\end{itemize}

The proof section intensively use the following equilibrum properties.

\subsection{Some equilibrium properties}

In the following, $\gamma$ is a continuous, strictly increasing and strictly convex function such that $\gamma(0)=0$.
\begin{lemma}\label{prop:3claims}~~\\
Suppose that in the game with $n$ players, references $\textbf{r}$ and cost function $\gamma$, an equilibrium $\textbf{x}$ is played:
\begin{enumerate}
\item Players are either single or paired: for any player $i$ we have
$$Card(x_i,\textbf{x}) \leq 2$$ 
\item If two players are paired, the quantity of consumers buying to their locations from the left is equal to the quantity of consumers buying from the right: 
$$Card(x_i,\textbf{x}) = 2 \Rightarrow q_i^{\ell}(\textbf{x})=q_i^{r}(\textbf{x})$$ 
\item The quantity of consumers that a non-peripheral player attracts does not change when he deviates to a location between the same neighbors (also when he is paired). Formally: $y \mapsto q_i(y,\textbf{x}^{-i})$ is  constant with $y\in (x_i^{\ell},x_i^r)$ for a non-peripheral player $i$.
\item The quantity of consumers that a paired, left-peripheral and not right-peripheral player attracts does not change when he deviates to his right, as long as he doesn't go over his right neighbor: $y \mapsto q_i(y,\textbf{x}^{-i})$ is a constant function on the interval $[x_i,x_i^r)$. The same statement holds for a right-peripheral player deviating to his left.
\item A left-peripheral players which is not right-peripheral is located on the right of his reference: $r_i \leq x_i$ for $i$ left and not right-peripheral (resp. $r_j \geq x_j$ for $j$ right and not left-peripheral).
\item If the reference profile is sorted, then equilibrium is also sorted: $r_1\leq \dots \leq r_n \Rightarrow x_1 \leq \dots \leq x_n$.
\end{enumerate}
\end{lemma}

\begin{proof} Here we choose to avoid formalism when possible as the results are intuitive.\\
(1) Suppose that three players select the same location. Each of them gets a third of the clientele. If one of these players makes an infinitesimal move to the left, he loses consumers from the right but does not share consumers from the right anymore. Either left of right consumers is larger than a third of the total clientele, so either an infinitesimal deviation to the left or to the right improves the quantity of attracted consumers. Because the mentioned deviations are infinitesimal and the cost function is continuous, we described a profitable deviation. The same argument applies with any number of players $n \geq 3$.\\
(2) Suppose that two players select the same location and that they attract more consumers from the left ($q^{\ell}>q^r$). If a player makes an infinitesimal move to the left, he looses the right consumers but does not share the left consumers anymore. Again, because $\gamma$ is continuous, we described a profitable deviation.\\
(3) The statement is obvious if the player is single, as his clientele only depends on the distance between his left and right neighbors, which is constant. In the case where the player is paired, he attracts half the consumers buying from the left and half the consumers buying from the right. At equilibrium these two quantities are equal (Claim 2 above). A deviation to the left would change his right neighbor but not the quantity of attracted consumers.\\
(4) Suppose that a player is paired, left-peripheral and not right-peripheral. 
Again, a deviation to the right would change his left neighbor but not the total quantity of attracted consumers. This statement does not hold for a deviation to the left where he attract less consumers.\\
(5) If $x_i < r_i$ for $i$ left and not right-peripheral player then player $i$ has a profitable deviation to his right. It also holds when $i$ is paired due to Claim 4 above.\\
(6) Suppose that there exists $i<j$ such that $x_j<x_i$. We have either $x_i \leq r_j$ or $r_j < x_i$. Suppose first that $x_i \leq r_j$ so that $x_j<x_i \leq r_j$. In this situation Player $j$ may be non-peripheral (paired or not) or left-peripheral (paired or not).\\
- If Player $j$ is non-peripheral (paired or not), Claim 3 provides a profitable deviation within the same neighborhood, towards his reference.\\
- If Player $j$ is left-peripheral and paired, Claim 4 provides a profitable deviation to his right.\\
- If Player $j$ is left-peripheral and not paired, a small deviation to his right increases his clientele and decreases his cost.\\
Suppose now that $r_j < x_i$, so that $r_i \leq r_j < x_i $ and $x_j < x_i$. In this situation Player $i$ may be non-peripheral (paired or not) or right-peripheral (paired or not).\\
- If Player $i$ is non-peripheral (paired or not), Claim 3 above gives us that he would have a profitable local deviation towards his reference.\\
- If Player $i$ is right-peripheral (paired or not), he has a profitable deviation towards his reference as it lowers his cost while increasing (or maintaining if paired) his quantity of attracted resources.
\end{proof}

An immediate consequence of the previous lemma is the following strong equilibrium property:

\begin{lemma}\label{le:non_periph_at_ref}
Suppose that in the game with $n$ players, references $\textbf{r}$ and cost function $\gamma$, an equilibrium $\textbf{x}$ is played: if Player $i$ is non-peripheral then $x_i=r_i$.
\end{lemma}

\begin{proof}
Suppose that $x_i \neq r_i$ for a non-peripheral player $i$. Claim 3 of Lemma \ref{prop:3claims} tells that $q_i$ is constant in the interval $(x_i^{\ell},x_i^r)$ (also when $i$ is paired). Therefore, a deviation towards his reference within the same neighborhood is profitable as it reduces the costs.\\
\end{proof}

\subsection{Proof of Proposition \ref{prop:2players}.}\label{prof:prop2players}
Consider the duopoly competition with $\gamma(d)=cd^2$. We suppose, without loss of generality that $r_1\leq r_2$.\\ 

\underline{We first investigate the differentiated equilibrium.}
~~\\
Suppose that $r_1+\frac{1}{4c} < r_2-\frac{1}{4c}$ and that $\textbf{x}^{\star}=(x_1,x_2)$ is an equilibrium. Lemma \ref{prop:3claims} (claim 6) gives that $x_1\leq x_2$. We now show that $x_1 < x_2$. The same Lemma (claim 2) proves that if $x_1=x_2$ then $x_1=x_2=\frac{1}{2}$. Suppose ad absurdum that $\textbf{x}^{\star}=(\frac{1}{2},\frac{1}{2})$ is an equilibrium.

\begin{itemize}
\item If $r_1+\frac{1}{4c}<\frac{1}{2}$, consider Player 1 deviating towards $r_1<\frac{1}{2}$. The equilibrium condition gives that $g_1((\frac{1}{2} - \epsilon,\frac{1}{2}),\textbf{r})-g_1(\textbf{x}^{\star},\textbf{r}) = 2c\epsilon(\frac{1}{2}-r_1)-c\epsilon^{2}-\frac{\epsilon}{2} \leq 0$. Dividing by $\epsilon>0$ and taking $\epsilon \rightarrow 0$ leads to $\frac{1}{2}\leq\frac{1}{4c}+r_1$ which is in contradiction with our hypothesis.

\item Symmetrically, if $\frac{1}{2}\leq r_1+\frac{1}{4c}$, consider Player 2 deviating towards $r_2$. We have: $g_2((\frac{1}{2},\frac{1}{2}+\epsilon),\textbf{r})-g_2(\textbf{x}^{\star},\textbf{r}) = 2c\epsilon(r_2-\frac{1}{2})-c\epsilon^{2}-\frac{\epsilon}{2} \leq 0$. Dividing by $\epsilon>0$ and taking $\epsilon \rightarrow 0$ leads to $r_2-\frac{1}{4c}\leq \frac{1}{2}$. Because $r_1+\frac{1}{4c} < r_2-\frac{1}{4c}$, the previous inequality contradicts with our hypothesis.
\end{itemize}

We proved that $x_1< x_2$. 
In this case, the payoff functions are given by $g_1(x_1,x_2)= \frac{x_1+x_2}{2}-c(x_1-r_1)^2$ and $g_2(x_1,x_2)= 1-\frac{x_1+x_2}{2}-c(r_2-x_2)^2$. The first order conditions give the optimal location for each player since their payoff functions are concave. 
We obtain $x_1=r_1+\frac{1}{4c}$ and $x_2=r_2-\frac{1}{4c}$.
We have proved that if $r_1+\frac{1}{4c} < r_2-\frac{1}{4c}$ then $(r_1+\frac{1}{4c},r_2-\frac{1}{4c})$ is the unique equilibrium candidate.\\

To conclude that this profile is indeed an equilibrium we have to check that players have no incentive to deviate on the other side of their opponent (which would change the expression of the payoff function as $x_1>x_2$). We now prove that such deviations are not profitable when inequalities $(E1)$ and $(E2)$ are satisfied. 

The first of the two inequalities below translates the fact to Player 1 has no incentive to deviate to location $x_2+\epsilon$ as $\epsilon$ goes to zero positively (this criteria is both necessary and sufficient since his payoff would decrease with $\epsilon>0$). Symmetrically, the second inequality translates the fact to player 2 has no incentive to deviate to location $x_1-\epsilon$ as $\epsilon$ goes to zero positively.
$$g_1(\textbf{x}^{\star},\textbf{r}) = \frac{r_1+r_2}{2}-\frac{1}{16c} \geq 1-(r_2-\frac{1}{4c})-c(r_2-\frac{1}{4c}-r_1)^{2}$$
and 
$$g_2(\textbf{x}^{\star},\textbf{r}) = 1-\frac{r_1+r_2}{2}-\frac{1}{16c} \geq (r_1+\frac{1}{4c})-c(r_2-\frac{1}{4c}-r_1)^{2}$$
Those two inequalities simplify to:
$$c(r_2-r_1)^{2}+r_1+r_2-1-\frac{1}{4c}\geq 0$$
$$c(r_2-r_1)^{2}-(r_1+r_2)+1-\frac{1}{4c}\geq 0$$ 
which are $(E1)$ and $(E2)$ in Proposition \ref{prop:2players}.\\
Note that we also proved that player 1 has no incentive to deviate to $x_1=x_2$. Indeed, we have that $u_2(x_2,x_2)=\frac{1}{2}-c(x_2-r_2)^2= \displaystyle \lim_{\epsilon \downarrow 0} \frac{u_2(x_2-\epsilon,x_2)+u_2(x_2+\epsilon,x_2)}{2}$ and both terms in the numerator are smaller than $u_2(x_1,x_2)$ as we prove that is it not profitable for Player 1 to deviate to $x_2-\epsilon$ nor to $x_2+\epsilon$ for any $\epsilon>0$.\\
~~\\
\underline{We now investigate the undifferentiated equilibrium.}
~~\\

Suppose now that $r_2-\frac{1}{4c} \leq r_1+\frac{1}{4c}$ and that $\textbf{x}^{\star}=(x_1,x_2)$ is an equilibrium.

To prove that $(x_1,x_2)=(\frac12,\frac12)$, it is sufficient to show that $x_1=x_2$ (Lemma \ref{prop:3claims}, claim 2).

Suppose, ad absurdum, that $x_1<x_2$. In this case the payoffs are given by $g_1(x_1,x_2)= \frac{x_1+x_2}{2}-c(x_1-r_1)^2$ and $g_2(x_1,x_2)= 1-\frac{x_1+x_2}{2}-c(r_2-x_2)^2$.
The equilibrium conditions give that $ \frac{\partial g_1(x_1+\epsilon,x_2)}{\partial \epsilon} <0$ for $\epsilon$ small enough, which simplifies to $\frac{1}{4c}+r_1\leq x_1$ when $\epsilon \rightarrow 0$. It also gives that $ \frac{\partial g_2(x_1,x_2-\epsilon)}{\partial \epsilon} <0 $ for $\epsilon$ small enough which simplifies to $x_2 \leq r_2-\frac{1}{4c}$ when $\epsilon \rightarrow 0$. We would then have $\frac{1}{4c}+r_1\leq x_1<x_2 \leq r_2-\frac{1}{4c}$, which contradicts our hypothesis.\\
~~\\
We proved that if $r_2-\frac{1}{4c} \leq r_1+\frac{1}{4c}$, the unique equilibrium candidate is $(\frac12,\frac12)$. 
It remains to prove that if $r_2-\frac{1}{4c}\leq\frac{1}{2}\leq r_1+\frac{1}{4c}$ then $(\frac12,\frac12)$ is an equilibrium.\\

Suppose that $r_2-\frac{1}{4c}\leq \frac{1}{2} \leq r_1+\frac{1}{4c}$. Suppose first that $r_1 \leq \frac{1}{2}$.\\
First, notice that player 1 has no profitable deviations in the interval $(\frac{1}{2},1]$ as it decreases his clientele and increases his costs. Furthermore, any deviations in $[r_1,\frac{1}{2}]$ would give Player $1$ a payoff $g_1((\frac{1}{2}-\epsilon,\frac{1}{2}),\textbf{r})=g_1(\textbf{x},\textbf{r})+2c\epsilon(\frac{1}{2}-r_1)-c\epsilon^{2}-\frac{\epsilon}{2}$ (with $0<\epsilon<\frac{1}{2}-r_1$). Because $\frac{1}{2}-\frac{1}{4c} \leq r_1$, we have $2c\epsilon(\frac{1}{2}-r_1)-c\epsilon^{2}-\frac{\epsilon}{2}\leq 0$ and the deviation is not profitable. Finally, any deviation in the interval $[0,r_1)$ gives a payoff strictly smaller than a deviation in $r_1$ which is not profitable. The same arguments hold for $r_1>\frac{1}{2}$ and for Player $2$ so we proved that $(\frac12,\frac12)$ is an equilibrium.

\subsection{Proof of Proposition \ref{prop:theta}}\label{proof:prop_theta}



We compute the area of the set $\mathcal{E}=\{ \textbf{r}\in[0,1]^{2} \text{ for which the duopoly competition admits an equilibrium}\}$. The computation is made easier after a rotation of the axis: $(x,y)=\left(\frac{\sqrt{2}}{2}(r_1-r_2),\frac{\sqrt{2}}{2}(r_1+r_2)\right)$ and when we consider the non-equilibrium area noted $\mathcal{N}$. By symmetry, we only compute $\mathcal{B}$, the bottom part of $\mathcal{N}$, below the line $y=\frac{\sqrt{2}}{2}$. We deduce $\mathcal{E}=1-\mathcal{N}=1-2\mathcal{B}$.

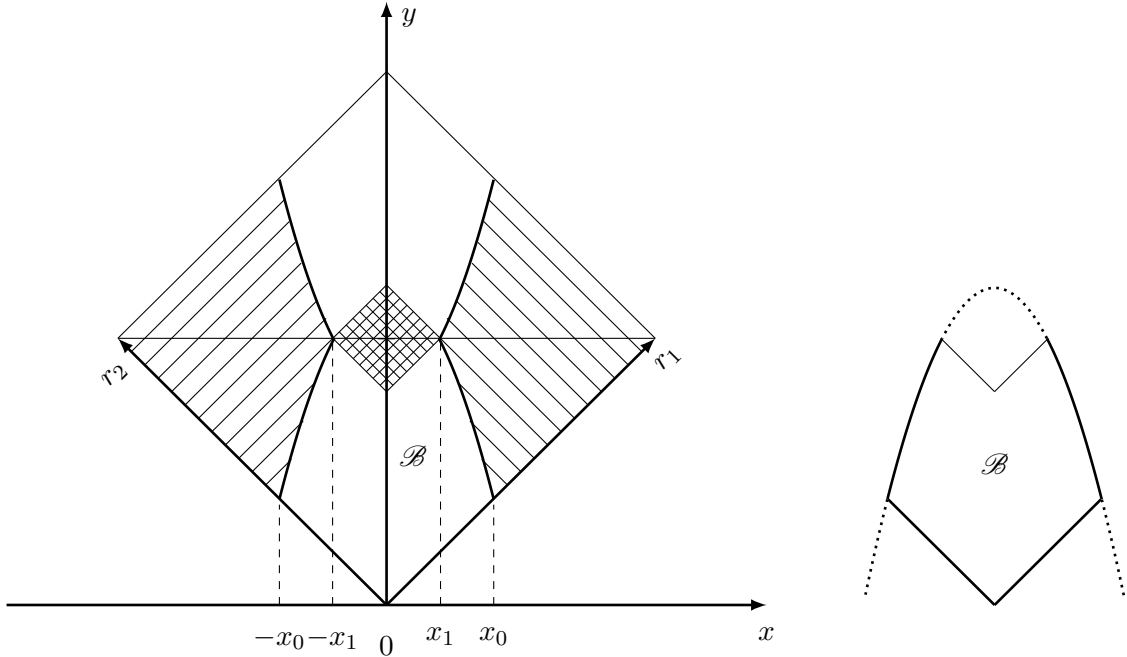
\begin{figure}[H]
\centering
\begin{tikzpicture}[>=latex]

\begin{scope}[xshift=0cm,yshift=0cm,rotate=45]

\draw[line width=1pt,->] (0,0) --(0,5);
\draw[line width=1pt,->] (0,0) --(5,0);
\draw[line width=0.4pt] (5,5) --(0,5);
\draw[line width=0.4pt] (5,0) --(5,5);

\draw[] (0,5) --(5,0);

\begin{scope}[xshift=1.87cm,yshift=1.87cm,rotate=135]

\draw[line width=1pt, smooth,samples=100,domain=-1.41:-0.69] plot(\x,1.41*\x*\x-1.56);
\draw[line width=1pt, smooth,samples=100,domain=0.69:1.41] plot(\x,1.41*\x*\x-1.56);

\end{scope}

\begin{scope}[xshift=3.11cm,yshift=3.11cm,rotate=315]

\draw[line width=1pt, smooth,samples=100,domain=-1.41:-0.69] plot(\x,1.41*\x*\x-1.56);
\draw[line width=1pt, smooth,samples=100,domain=0.69:1.41] plot(\x,1.41*\x*\x-1.56);

\end{scope}

\node[] at (2.5,2.5) {$|$};
\node[] at (2.5,2.5) {$-$};
\node[below=8pt] at (0,0) {$0$};
\node[below=3pt,rotate=45] at (5,-0) {$r_1$};
\node[below=3pt,rotate=45] at (-0.3,5) {$r_2$};

\draw[] (2.25,0) --(2.25,0.48);
\draw[] (2.5,0) --(2.5,0.95);
\draw[] (2.75,0) --(2.75,1.44);
\draw[] (3,0) --(3,2);
\draw[] (3.25,0) --(3.25,2.125);
\draw[] (3.5,0) --(3.5,2.21);
\draw[] (3.75,0) --(3.75,2.29);
\draw[] (4,0) --(4,2.41);
\draw[] (4.25,0) --(4.25,2.55);
\draw[] (4.5,0) --(4.5,2.69);
\draw[] (4.75,0) --(4.75,2.86);

\draw[] (0,2.25) --(0.48,2.25);
\draw[] (0,2.5) --(0.95,2.5);
\draw[] (0,2.75) --(1.44,2.75);
\draw[] (0,3) --(2,3);
\draw[] (0,3.25) --(2.125,3.25);
\draw[] (0,3.5) --(2.21,3.5);
\draw[] (0,3.75) --(2.29,3.75);
\draw[] (0,4) --(2.41,4);
\draw[] (0,4.25) --(2.55,4.25);
\draw[] (0,4.5) --(2.69,4.5);
\draw[] (0,4.75) --(2.86,4.75);

\draw[] (2,2) --(3,2);
\draw[] (2,3) --(3,3);
\draw[] (2,2) --(2,3);
\draw[] (3,2) --(3,3);

\draw[] (2.5,2) --(2.5,3);
\draw[] (2.12,2) --(2.12,3);
\draw[] (2.25,2) --(2.25,3);
\draw[] (2.375,2) --(2.375,3);
\draw[] (2.62,2) --(2.62,3);
\draw[] (2.75,2) --(2.75,3);
\draw[] (2.875,2) --(2.875,3);
\draw[] (2,2.5) --(3,2.5);
\draw[] (2,2.12) --(3,2.12);
\draw[] (2,2.25) --(3,2.25);
\draw[] (2,2.375) --(3,2.375);
\draw[] (2,2.62) --(3,2.62);
\draw[] (2,2.75) --(3,2.75);
\draw[] (2,2.875) --(3,2.875);
\end{scope}

\draw[line width=1pt,->] (0,0) --(0,8);
\draw[line width=1pt,->] (-5,0) --(5,0);

\node[below=3pt] at (5,-0.06) {$x$};
\node[below=3pt] at (+0.3,8.15) {$y$};

\draw[dashed] (0.71,0) --(0.71,3.55);
\draw[dashed] (-0.71,0) --(-0.71,3.55);

\draw[dashed] (1.41,0) --(1.41,1.33);
\draw[dashed] (-1.41,0) --(-1.41,1.33);

\node[below=3pt] at (0.71,-0.06) {$x_1$};
\node[below=3pt] at (1.41,-0.06) {$x_0$};

\node[below=3pt] at (-0.71,-0.06) {$-x_1$};
\node[below=3pt] at (-1.41,-0.06) {$-x_0$};

\node[below=3pt] at (0.35,2.35) {$\mathcal{B}$};

\node[below=3pt] at (8,2.25) {$\mathcal{B}$};

\begin{scope}[xshift=8cm,yshift=0cm,rotate=45]

\draw[line width=1pt] (0,0) --(0,1.98);
\draw[line width=1pt] (0,0) --(1.98,0);

\begin{scope}[xshift=1.87cm,yshift=1.87cm,rotate=135]

\draw[line width=1pt, smooth,samples=100,domain=-1.41:-0.69] plot(\x,1.41*\x*\x-1.56);
\draw[line width=1pt, smooth,samples=100,domain=0.69:1.41] plot(\x,1.41*\x*\x-1.56);

\draw[dotted, line width=1pt, smooth,samples=100,domain=-1.7:-1.41] plot(\x,1.41*\x*\x-1.56);
\draw[dotted, line width=1pt, smooth,samples=100,domain=1.41:1.7] plot(\x,1.41*\x*\x-1.56);

\draw[dotted, line width=1pt, smooth,samples=100,domain=-0.69:0.69] plot(\x,1.41*\x*\x-1.56);

\end{scope}

\draw[] (2,2) --(3,2);

\draw[] (2,2) --(2,3);

\end{scope}

\end{tikzpicture}
\caption{The equilibrium area $\mathcal{E}$ in hatched, the non-equilibrium area $\mathcal{E}$ in white. $\mathcal{B}$ is the bottom part of $\mathcal{N}$.}
\end{figure}

To compute $\mathcal{B}$, we compute the integral of the parabola between $-x_0$ and $x_0$ and we subtract the missing parts. 

\[ \mathcal{B} = \int_{-x_0}^{x_0} (-\sqrt{2}cx^{2}+\frac{\sqrt{2}}{2} +\frac{\sqrt{2}}{8c}) \,dx  - \Bigg(2\times\frac{x_0^{2}}{2}\Bigg) - \Bigg( \int_{-x_1}^{x_1} (-\sqrt{2}cx^{2}+\frac{\sqrt{2}}{2} +\frac{\sqrt{2}}{8c}) \,dx - 2\bigg(x_1(\frac{\sqrt{2}}{2}-\frac{\sqrt{2}}{4c})+\frac{x_1^{2}}{2}\bigg)   \Bigg)\]

which gives $\mathcal{A}=1-2\mathcal{B}=1-\frac{(2+4c)^{\frac{3}{2}}-6c-5}{6c^2}$. We finally have $\mathbb{P}(c)=\min\left(1,\mathcal{A}\right)=\min\left(1,1-\frac{(2+4c)^{\frac{3}{2}}-6c-5}{6c^2}\right)$.\\
~~\\
We now prove that $c \mapsto \mathbb{P}(c)$ is constant on $[0,\frac{1}{2}]$, decreasing on $[\frac{1}{2},\theta]$ and increasing on $[\theta,+\infty[$ for $\theta \simeq 1.1321$.

First, for $c\in (0,\frac{1}{2}]$, we obtain that $(\frac12,\frac12)$ is an equilibrium for every profile of reference so $\mathcal{P}(c)=1$. Furthermore, for $c>\frac12$, we find that $c \rightarrow \mathbb{P}(c)$ is a smooth function and that $\frac{d\mathbb{P}(c)}{dc}$ has the same sign as $f(c)=\sqrt{4c+2}(c+2)-(3c+5)$. We have $f'(c)=\sqrt{2}\sqrt{2c+1}+\frac{\sqrt{2}(c+2)}{\sqrt{2c+1}}-3$ and $f''(c)= \frac{3c\sqrt{2}}{(2c+1)^\frac{3}{2}}>0$. Because $f'$ is increasing and because $f'(\frac{1}{2})= \frac32>0$ we obtain that $f'$ is positive and therefore that $f$ is increasing.\\
~~\\
Because $f(\frac{1}{2})=-\frac32$ and that $\displaystyle\lim_{c\rightarrow+\infty} f(c)=+\infty$, the intermediate value theorem provides the existence of $\theta$ such that $f(x)\leq 0$ when $x \leq \theta$ and $f(x)\geq 0$ when $x \geq \theta$.\\

Solving $\frac{d\mathbb{P}(c)}{dc}=0$ gives a polynomial equation of degree $3$ that provides a unique real solution $\theta=\frac{1}{12} \left(3 \sqrt[3]{4 \sqrt{26}+73}+\sqrt[3]{1971-108 \sqrt{26}}-9\right) \simeq 1.321$.

\subsection{Proof of Proposition \ref{prop:2players_hetero}.}\label{proof:prop2players_hetero}

The proof is a straightforward generalization of Proposition \ref{prop:2players}: the quantity $\frac{1}{4c}$ that appears with homogeneous costs is replaced by either $\delta_1=\frac{1}{4c_1}$ or $\delta_2=\frac{1}{4c_2}$. When repeating the proof, there are two minor modifications:\\
(1) in the case where $r_2-\delta_2 \leq r_1+\delta_1$: in the homogeneous case, the necessary condition for $(\frac{1}{2},\frac{1}{2})$ to be an equilibrium is that $\frac{1}{2} \in [r_2-\delta,r_1+\delta]$ and this condition was sufficient to obtain that $\frac{1}{2} \in [r_1-\delta,r_1+\delta] \cap [r_2-\delta,r_2+\delta]$. This is not the case anymore with heterogeneous costs, so the necessary condition is now written with the two conditions $\frac{1}{2} \in [r_1-\delta_2,r_1+\delta_2]$ and $\frac{1}{2} \in [r_2-\delta_2,r_2+\delta_2]$.\\
(2) in the case where $r_1+\delta_1<r_2-\delta_2$: similarly to the homogeneous case, to insure that Player 1 has no incentive to deviate to the right of Player 2, and reciprocally that Player 2 has no incentive to deviate to the left of Player 1, we write the two conditions on the references $r_1,r_2$. Because the costs functions have changed, these conditions are now provided by inequalities $(E1'')$ and $(E2'')$.

\subsection{The triopoly}

\subsubsection{The equilibrium candidate.}\label{proof:candidate_triopoly}

Let $\textbf{r}=(r_1,r_2,r_3) \in [0,1]^{3}$ and $\gamma(d)=cd^2$. Without loss of generality we suppose that that $r_1 \leq r_2 \leq r_3$. Let $\textbf{x}^{*} = (x_1,x_2,x_3)$ be an equilibrium of the game. Remark that Claim 6 of Lemma \ref{prop:3claims} gives that $x_1\leq x_2 \leq x_3$.\\
~~\\
We start by providing a list of triopoly equilibrium properties and we then conclude on the equilibrium structure.\\
We show that if Player $1$ is not a far-left left player, he is necessarily paired with Player $2$ at equilibrium. Suppose ad absurdum that $r_1 + \delta \geq r_2$ and that $x_1<x_2$ ($\delta$ is defined in \ref{def:delta} and is equal to $\frac{1}{4c}$ in this context). Player $1$'s payoff is $g_1(\textbf{x}^{\star},\textbf{r})=\frac{x_1+x_2}{2}-c(x_1-r_1)^2$. Because this function is concave and because its derivative is null in $r_1+\delta$, it is then increasing on $[0, r_1 + \delta]$. We now prove that $x_2 \leq r_1+\delta$, which implies that Player $1$'s payoff is strictly increasing on $[0,x_2)$ and contradicts $x_1<x_2$. Suppose first that $x_2<x_3$, Lemma \ref{le:non_periph_at_ref} gives that $x_2=r_2$ so that $x_2 \leq r_1 + \delta$. Suppose now that $x_2=x_3$, both players are right-peripheral. We necessarily have $x_3=x_2\leq r_2$ (Lemma \ref{prop:3claims}, claim 5) and therefore $x_2 \leq r_1 + \delta$. We proved that if Player $1$ is not far-left he is necessarily paired with Player $2$. The same arguments hold for Player $3$ if he is not far-right.\\

A direct consequence is that if neither Player $1$ nor Player $3$ are far-left and far-right, there is no equilibrium as three players can not share the same location (Lemma \ref{prop:3claims}, claim 1).\\

We secondly show that, if a peripheral player is far-left, he is not paired. The equilibrium conditions imply that any deviation $\epsilon$ of Player $1$ towards his reference is not profitable, so that $g_1((x_1-\epsilon,x_2,x_3),\textbf{r})-g_1((x_1,x_2,x_3),\textbf{r})=-\frac{\epsilon}{2}+2c(x_1-r_1)\epsilon-c\epsilon^{2} \leq 0$. Dividing by $\epsilon>0$ and taking $\epsilon \rightarrow 0$ gives that that $x_1 \leq r_1 + \frac{1}{4c}$. Moreover, we have $r_1\leq r_2 \leq x_1=x_2$ (Claim 5 of the same Lemma) which gives a contradiction with our hypothesis that Player 1 is far left. The same arguments arise for the symmetric case of a far-right player.\\


Because Player $1$ is not paired when he is far-left, we show that $x_1=r_1+\delta$. Suppose that $r_1 + \delta < r_2$. In this situation Player $1$'s payoff is $g_1(\textbf{x}^{\star},\textbf{r})=\frac{x_1+x_2}{2}-c(x_1-r_1)^2$. The first and second order conditions gives us that his payoff is increasing with $x_1 \in [0, r_1 + \delta]$ and decreasing for $x_1 \in [r_1 + \delta,x_2)$. His payoff is maximized in $x_1 = r_1 + \delta$, which is the only position possible at equilibrium in $[0,x_2)$. The same arguments arise for the symmetric case of a far-right player.\\

In the equilibrium situation where $x_1<x_2<x_3$, Lemma \ref{le:non_periph_at_ref} gives that $x_2=r_2$. Finally, Claim 2 of Lemma \ref{prop:3claims} gives that if Players $1$ and $2$ are paired $x_1=x_2=\frac{x_3}{3}$.\\

We can now conclude on the structure of the triopoly equilibrium:
\begin{itemize}
\item First, if both Player $1$ and Player $3$ are far-left and far-right, then $x_1=r_1+\delta$, $x_2=r_2$, $x_3=r_3-\delta$. 
\item If Player $1$ is not far-left but Player $3$ is far-right, Player $1$ and $2$ are paired and Player $3$ is alone on $x_3=r_3-\delta$. Player $1$ and $2$ are positioned on $\frac{x_3}{3}=\frac{r_3-\delta}{3}$. The symmetric situations leads to $x_1=r_1+\delta$ and $x_2=x_3=\frac{r_1+\delta+2}{3}$.
\item If both Player $1$ and Player $3$ are not far-left nor far-right, then there is no equilibrium.
\end{itemize}

\subsubsection{Proof of Proposition \ref{prop:58noNE}}\label{proof:58noNE}

Suppose first that $c < \frac{5}{8}$, that is $\delta=\frac{1}{4c} > \frac{2}{5}$. We first prove that there exists no equilibrium with a far-left and a far-right players. Indeed, in this case: $x_1=r_1+\delta  > \frac{2}{5}$ and $x_3=r_3-\delta < \frac{3}{5}$. Therefore, Player $2$'s payoff is strictly less than $\frac{1}{10}$ and deviation to $x_1-\epsilon$ for $\epsilon$ small enough is profitable as Player $2$'s deviation payoff is at at least $\frac{2}{5}-c(\frac{1}{5})^2 > \frac{2}{5}-\frac{1}{30}>\frac{1}{5}$.\\ 

We now prove that there exists no equilibrium with a far-left player and no far-right player (the other case is symmetric). In this case, $x_1=r_1+\delta > \frac{2}{5}$. According to the equilibrium candidate described in \ref{proof:candidate_triopoly}, Player $2$ and $3$ share the location $\frac{r_1+\delta+2}{3}$. Player $2$'s payoff decreases with $r_1$ and is at most $\frac{1}{5}$. We now prove that Player $2$ has a profitable deviation to the left of Player $1$. In this location, Player $2$ attracts a quantity of consumers arbitrary close to $r_1+\delta \geq \delta > \frac25$. Because the cost function is convex the deviation is the least profitable when $r_1=0$ and $r_2=1$ and the deviation increases the costs by at most $c(\frac{3}{5})^2-c(\frac{1}{5})^2 < \frac{1}{5}$. When Player $2$ deviates to the left of Player $1$, his payoff change is strictly larger than $\frac{2}{5}-\frac{1}{5}-\frac{1}{5}=0$, and is therefore profitable. The exact same argument applies to Player $3$. We proved that there exists no equilibrium.\\

Suppose that $c \geq \frac{5}{8}$. There exits at most one equilibrium, as proved in the previous subsection where we exhibit the unique equilibrium candidate.\\

As long as $r_1<r_2<r_3$, there exists $c$ large enough so that Player $1$ and $3$ are far left and far right. The unique equilibrium candidate is therefore $(r_1+\frac{1}{4c},r_2,r_3-\frac{1}{4c})$. For $c$ large enough, the cost of any deviation for any player is strictly larger than $1$ which the maximal payoff a player can obtain in the game, so deviations are not profitable.

\subsection{Proof of Proposition \ref{prop:triopoly_symmetric}}\label{proof:triopoly_symmetric}

Consider the triopoly competition where $\textbf{r}=(r_1,\frac{1}{2},1-r_1)$ and $\gamma(d)=cd^2$.\\

If Player $1$ is not far-left then by symmetry Player $3$ is not far-right and Proposition \ref{proof:candidate_triopoly} gives that there is no equilibrium. At equilibrium, we must have Player $1$ and $3$ being respectively far-left and far-right. 

According to \ref{proof:candidate_triopoly}, the unique possible equilibrium is $(r_1+\delta,\frac12,1-r_1-\delta)$. We start by providing necessary conditions. Because Player $2$ has no profitable deviation to the left of Player $1$, we have:
$$\frac{1-r_1-\delta-r_1-\delta}{2} \geq r_1+\delta-c(\frac{1}{2}-r_1-\delta)^2$$ which simplifies to the second order equation $$c^{2}r_1^{2}-r_1(c^2+\frac{3}{2}c)+(\frac{c^2}{4}+\frac{c}{4}-\frac{7}{16})\geq 0$$  For the equilibrium candidate to be an equilibrium, $r_1$ should satisfy $r_1 \notin \left[\frac{3+2c-2\sqrt{2c+4}}{4c},\frac{3+2c+2\sqrt{2c+4}}{4c}\right]$. Because $\frac{3+2c+2\sqrt{2c+4}}{4c}>\frac12$ while $r_1<\frac{1}{2}$, we can simplify the condition to $r_1\leq \frac{3+2c-2\sqrt{2c+4}}{4c} =:\phi(c)$. \\

Now suppose that $\textbf{x}=(r_1+\frac{1}{4c},\frac{1}{2},1-r_1-\frac{1}{4c})$ and that $r_1\in [0,\phi(c)]$. We prove that $\textbf{x}$ is an equilibrium. 

First notice that players have no profitable deviation within their neighborhoods. Indeed, the first and second order conditions provide that the payoff is maximized when $x_1=r_1+\delta$ and $x_3=r_3-\delta$, and Claim $3$ Lemma \ref{prop:3claims} gives that Player $2$ prefers to locate at $r_2$ in this neighborhood. 

Moreover, players have no deviations to a different neighborhood: a deviation of Player $1$ to $x \in (\frac{1}{2},1]$ provides a smaller payoff than a deviation to the symmetric location $1-x \in [0,\frac12)$ which is not a profitable deviation, as it belongs to the neighborhood as $x_1$. It is not profitable to deviate to $x_2$ neither as such a deviation is less profitable than a deviation to $x_2-\epsilon$ for $\epsilon$ small enough, which is not profitable. Following the above computations, Player $2$ has no incentive to deviate to the left of Player $1$ or to the right of $3$ as long as $r_1 \leq \phi(c)$. Finally, Player $2$ has no incentive to deviate to $x_1$ as neither $x_1-\epsilon$ nor $x_1+\epsilon$ are profitable deviations as $\epsilon \rightarrow 0$. The same arguments hold for $x_3$. 

\subsection{Proof of section \ref{se:general_case}}

\subsubsection{Proof of Proposition \ref{thm:equilibrium_description}}\label{proof:general}

Suppose that $\textbf{x}$ is an equilibrium in the game with references $\textbf{r}$ and suppose, without loss of generality, that $r_1\leq \dots \leq r_n$.\\

As a consequence of Lemma \ref{prop:3claims} (Claim 6), Players $1$ and $n$ are peripheral. Players $2$ and $n-1$ can be peripheral or non-peripheral. Because at most $2$ players can share the same location at equilibrium (Claim 1, Proposition \ref{prop:3claims}), players $3$, $4$, \dots, $n-2$ are non-peripheral. For these players, Lemma \ref{le:non_periph_at_ref} claims that $x_i=r_i$.\\

Suppose that Player $1$ is a far-left player, i.e. that $r_1 + \delta < r_2$. We prove that $x_1<x_2$. Suppose, ad absurdum, that $x_1=x_2$. Using Lemma \ref{prop:3claims}, we have $q_1^{\ell}(x)=q_1^r(x)$, and because their right neighbor is the non-peripheral player $3$ located at $x_3=r_3$, we have $x_1=x_2=\frac{r_3}{3}$. 

On the one hand, an infinitesimal deviation to his left is not profitable to Player $1$, so we have that $\gamma'(\frac{r_3}{3}-r_1)  \leq \frac{1}{2}$. On the other hand we have $r_1+\delta < r_2 \leq \frac{r_3}{3}$, where the second inequality comes from the fact that Player $2$ has no profitable infinitesimal deviation to his right. Because $\gamma'$ is strictly increasing, we have $\gamma'(\frac{r_3}{3}-r_1) > \gamma'(\delta)=\frac{1}{2}$ and we face a contradiction. Player $1$ is not paired with Player $2$.

Because Player $1$ is peripheral and single, his payoff is given by $u_i(\textbf{x})= \frac{x_1+x_2}{2}-\gamma(x_1-r_1)$ in the interval $[0,x_2)$ and the first and second order conditions give that his payoff is maximized in $x_1=r_1+\delta$. The argument is symmetric for Player $n$.\\

Suppose now that Player $1$ is not a far left-player, i.e. that $x_1+\delta \geq r_2$. If $x_1<x_2$, his payoff is given by $u_i(\textbf{x})= \frac{x_1+x_2}{2}-\gamma(x_1-r_1)$ and is therefore strictly increasing with $x_1 \in [0,x_2)$. We conclude that Players $1$ and $2$ are paired. This is only possible if $x_1=x_2=\frac{r_3}{3}$ (Lemma \ref{prop:3claims}, claim 2). The argument is symmetric for players $n-1$ and $n$.

\subsubsection{Proof of Proposition \ref{prop:sufficient_general}}
We first show that conditions $(1)$, $(2)$ and $(3)$ are necessary for the described profile to be an equilibrium.

Suppose that condition $(1)$ doesn't hold, i.e. Player $1$ is a far-left player but $\frac{r_3}{3} \notin [r_2,r_1+\delta]$.\\
- If $\frac{r_3}{3}<r_2$, then player $2$ has a profitable deviation to his right: the quantity that Player $2$ attracts is constant on the interval $[\frac{r_3}{3},r_3)$ (Claim 4, Proposition \ref{prop:3claims}) but his reference cost decreases on the interval $[\frac{r_3}{3},r_2]$.\\
- If $\frac{r_3}{3}>r_1+\delta$, then Player $1$ has a profitable deviation to his left: by making an infinitesimal move to his left, he obtains a marginal loss in the quantity of consumers which is equal to $\frac{1}{2}$ while the marginal diminution of his reference cost is $\gamma'(x_1-r_1)=\gamma'(\frac{r_3}{3}-r_1)>\gamma'(r_1+\delta-r_1)=\gamma'(\delta)=\frac{1}{2}$.\\
~~\\
The case where condition $(2)$ doesn't hold is symmetric.\\
~~\\
Suppose that condition $(3)$ doesn't hold, i.e. there exists $i$, $j$ such that $\Delta_i^j(\textbf{x}^*)<0$. Then, by definition of $\Delta_i^j(\textbf{x}^*)$, Player $i$ has a profitable deviation by playing $x_j + \epsilon$ for $\epsilon$ small enough if $i<j$, or by playing $x_j - \epsilon$ for $\epsilon$ small enough if $i>j$. \\
~~\\
We now prove that conditions $(1)$, $(2)$ and $(3)$ are sufficient. We suppose that $\textbf{x}^*$ is the equilibrium candidate described in Proposition \ref{thm:equilibrium_description}, that the three conditions hold and we prove that $x^*$ is an equilibrium. \\

Take any player $i\in \{ 3,\dots, n-2 \}$. He has no profitable deviation within his neighborhood $[x_i^{l},x_i^{r}]$ where his payoff is maximized at $x_i=r_i$.\\
We now prove that player $i$ has no profitable deviation towards a different neighborhood. Consider any player $j$ such that $x_j<x_i$. The deviation payoff of Player $i$ increases in the interval $(x_j^{\ell},x_j)$ as the quantity of attracted consumers is constant and the reference cost is decreasing. However, condition $(3)$ gives that $u_i(\textbf{x}) \geq \lim\limits_{\substack{\epsilon \to 0 \\ \epsilon>0}} u_i(x_j-\epsilon,\textbf{x}_{-i})$, which proves that there is no profitable deviation in the interval $(x_j^{\ell},x_j)$. Note that a deviation to $x_j$ is not profitable neither as it gives Player $i$ a payoff $u_i(x_j,x_{-i})= \lim\limits_{\substack{\epsilon \to 0 \\ \epsilon>0}}\frac{u_i(x_j+\epsilon,\textbf{x}_{-i})+u_i(x_j-\epsilon,\textbf{x}_{-i})}{2}$, both terms in the numerator being smaller than $u_i(\textbf{x})$ due to condition $(3)$. Because the same argument applies to every player $j$ such that $x_j<x_i$, we proved that Player $i$ has no incentive to deviate to any location in $[0,x_i^{\ell}]$. A symmetric argument applies with any player $j$ such that $x_j>x_i$ and we proved that Player $i$ has no profitable deviation towards a different neighborhood.\\

If Player $1$ is far-left and $\textbf{x}^*$ is played, Player $2$ is not peripheral so he has no profitable deviation.
Suppose now that Player $1$ is not far-left, so that Player $2$ is peripheral. We prove that he has no profitable deviation neither. Condition (1) gives $x_2=\frac{r_3}{3} \in [r_2,r_1+\delta]$.\\
- Player $2$ has no profitable deviations within $(x_2,x_3]$ due to Lemma \ref{prop:3claims}, Claim 4.\\
- Player $2$ has no profitable deviations within $[r_2,x_2)$ neither: for any deviation $x_2-\epsilon$ in this interval, we have $u_2(x_2-\epsilon,\textbf{x}_{-i})-u_2(\textbf{x}))=-\frac{\epsilon}{2}-c\epsilon^{2}+2c\epsilon(\frac{r_3}{3}-r_2)$. Yet, since $\frac{r_3}{3}\leq r_1+\delta \leq r_2+\delta$, we have $-\frac{1}{2}+2c(\frac{r_3}{3}-r_2)\leq 0$ and therefore $-\frac{\epsilon}{2}-c\epsilon^{2}+2c\epsilon(\frac{r_3}{3}-r_2)<0$. The deviation is not profitable.\\
- Player $2$ has no profitable deviations within $[0,r_2)$ as they are even less profitable than a deviation in $[r_2,x_2^{*})$.\\
We conclude that Player $2$ doesn't have any profitable deviation in his neighborhood $[0,x_3)$.\\
Moreover, the same argument than before applies for global deviation: condition (3) implies that any deviation in a different neighborhood is not profitable neither. The same arguments hold regarding Player $n-1$.\\
~~\\
If Player $1$ is far-left, he is the only peripheral player. He has no local deviations, as his payoff is maximum in $[0,x_2)$ for $x_1=r_1+\delta=x_1^{*}$, due to the first and second order conditions. Again, conditions $(3)$ implies that any deviation in a different neighborhood is not profitable.\\
If Player $1$ is far-left, he is paired with Player $2$. In this case, the same arguments used for Player $2$ above applies and Player $1$ has no profitable deviation. The same applies to Player $n$.

\subsection{Proof of Proposition \ref{prop:regular_n_players}}\label{proof:uniform_general}

We consider the case of $n\geq5$ players where $\gamma(d)=cd^2$ and $\bold{r}=(\frac{1}{n+1},\frac{2}{n+1},...,\frac{n}{n+1})$.\\

Suppose that $\textbf{x}$ is an equilibrium. Using Proposition \ref{thm:equilibrium_description}, we have $x_i=r_i=\frac{i}{n+1}$ for every $i\in\{3, \dots, n-2 \}$. If Player $1$ is not far-left, Proposition \ref{prop:sufficient_general} gives that $\frac{r_3}{3}=\frac{1}{n+1}\in[r_2,r_1+\delta]=[\frac{2}{n+1},\frac{1}{n+1}+\delta]$, which is not the case here. Therefore, there is no equilibrium when Player $1$ is not far-left, which is the case when $c < \frac{n+1}{4}$.\\
~~\\
The existence of an equilibrium implies $c \geq \frac{n+1}{4}$. In this case players $1$ and $n$ are respectively far-left and far-right. The unique equilibrium candidate is $\bold{x^{\star}}=(r_1+\delta,r_2,r_3,...,r_{n-1},r_{n}-\delta)$. A necessary condition for this profile to be an equilibrium is that Player $2$ does not profitably deviate on the left of Player $1$, that is:
$$r_1+\delta -c(r_2-r_1-\delta)^{2} \leq \frac{r_3-r_1-\delta}{2}$$
which simplifies to $c \geq \frac{1+\sqrt{6}}{4}(n+1)$. Note that this constraints is harder to satisfy than the constraint $c \geq \frac{n+1}{4}$.\\

We now prove that $\bold{x^{\star}}=(r_1+\delta,r_2,r_3,...,r_{n-1},r_{n}-\delta)$ is an equilibrium when $c \geq \frac{1+\sqrt{6}}{4}(n+1)$. Repeating standard arguments, no player has an incentive to deviate within his neighborhood. We now prove that Players have no profitable deviation to a different neighborhood.

Any player $i\in \{3,\dots,n-2\}$ has a payoff equal to $\frac{1}{n+1}$. Therefore, it is not profitable to deviate in the interval between two non-peripheral players. Moreover, because $c \geq \frac{1+\sqrt{6}}{4}(n+1)$, Player $2$ does not benefit from deviating to the t of Player $i$. Such a deviation is even less profitable for a player $i \in \{3,\dots,n-2\}$. The same argument applies for a deviation on the right of player $n$.
 
Players $2$ and $n-1$ have no incentives to relocate respectively on the left of Player $1$ or on the right of Player $n$ because $c \geq \frac{1+\sqrt{6}}{4}(n+1)$. Any other deviation is even less profitable as the interval are smaller and the reference costs are higher.

Players $1$ and $n$ have no incentives to deviate anywhere as any deviation both reduces their clientele and increases their references costs.

\bibliographystyle{plainnat}
\bibliography{bib_reputation}

\begin{thebibliography}{24}
\providecommand{\natexlab}[1]{#1}
\providecommand{\url}[1]{\texttt{#1}}
\expandafter\ifx\csname urlstyle\endcsname\relax
  \providecommand{\doi}[1]{doi: #1}\else
  \providecommand{\doi}{doi: \begingroup \urlstyle{rm}\Url}\fi

\bibitem[DRO(2014)]{DROUVELIS201486}
Political motivations and electoral competition: Equilibrium analysis and
  experimental evidence.
\newblock \emph{Games and Economic Behavior}, 83:\penalty0 86--115, 2014.
\newblock ISSN 0899-8256.

\bibitem[Brenner(2010)]{brenner2010location}
Steffen Brenner.
\newblock Location (hotelling) games and applications.
\newblock \emph{Wiley Encyclopedia of Operations Research and Management
  Science}, 2010.

\bibitem[Callander(2008)]{callander2008political}
Steven Callander.
\newblock Political motivations.
\newblock \emph{The Review of Economic Studies}, 75\penalty0 (3):\penalty0
  671--697, 2008.

\bibitem[Correia-da Silva and Pinho(2011)]{correia2011costly}
Jo{\~a}o Correia-da Silva and Joana Pinho.
\newblock Costly horizontal differentiation.
\newblock \emph{Portuguese Economic Journal}, 10\penalty0 (3):\penalty0
  165--188, 2011.

\bibitem[d'Aspremont et~al.(1979)d'Aspremont, Gabszewicz, and
  Thisse]{d1979hotelling}
Claude d'Aspremont, J~Jaskold Gabszewicz, and J-F Thisse.
\newblock On hotelling's" stability in competition".
\newblock \emph{Econometrica: Journal of the Econometric Society}, pages
  1145--1150, 1979.

\bibitem[Davis et~al.(1970)Davis, Hinich, and Ordeshook]{davis1970expository}
Otto~A Davis, Melvin~J Hinich, and Peter~C Ordeshook.
\newblock An expository development of a mathematical model of the electoral
  process.
\newblock \emph{The American Political Science Review}, 64\penalty0
  (2):\penalty0 426--448, 1970.

\bibitem[De~Palma et~al.(1987)De~Palma, Ginsburgh, and Thisse]{de1987existence}
Andr{\'e} De~Palma, Victor Ginsburgh, and Jacques-Francois Thisse.
\newblock On existence of location equilibria in the 3-firm hotelling problem.
\newblock \emph{The Journal of Industrial Economics}, pages 245--252, 1987.

\bibitem[Downs(1957)]{downs1957economic}
Anthony Downs.
\newblock An economic theory of political action in a democracy.
\newblock \emph{Journal of political economy}, 65\penalty0 (2):\penalty0
  135--150, 1957.

\bibitem[Eaton and Lipsey(1975)]{eaton1975principle}
B~Curtis Eaton and Richard~G Lipsey.
\newblock The principle of minimum differentiation reconsidered: Some new
  developments in the theory of spatial competition.
\newblock \emph{The Review of Economic Studies}, 42\penalty0 (1):\penalty0
  27--49, 1975.

\bibitem[Eaton and Schmitt(1994)]{eaton1994flexible}
B~Curtis Eaton and Nicolas Schmitt.
\newblock Flexible manufacturing and market structure.
\newblock \emph{The American Economic Review}, pages 875--888, 1994.

\bibitem[Fredriksson et~al.(2011)Fredriksson, Wang, and
  Mamun]{fredriksson2011politicians}
Per~G Fredriksson, Le~Wang, and Khawaja~A Mamun.
\newblock Are politicians office or policy motivated? the case of us governors'
  environmental policies.
\newblock \emph{Journal of Environmental Economics and Management}, 62\penalty0
  (2):\penalty0 241--253, 2011.

\bibitem[Hotelling(1990)]{hotelling1990stability}
Harold Hotelling.
\newblock Stability in competition.
\newblock In \emph{The collected economics articles of Harold Hotelling}, pages
  50--63. Springer, 1990.

\bibitem[Kartik and McAfee(2007)]{kartik2007signaling}
Navin Kartik and R~Preston McAfee.
\newblock Signaling character in electoral competition.
\newblock \emph{American Economic Review}, 97\penalty0 (3):\penalty0 852--870,
  2007.

\bibitem[Kishihara and Matsubayashi(2020)]{kishihara2020product}
Hiroki Kishihara and Nobuo Matsubayashi.
\newblock Product repositioning in a horizontally differentiated market.
\newblock \emph{Review of Industrial Organization}, 57\penalty0 (3):\penalty0
  701--718, 2020.

\bibitem[Lambertini(1997)]{lambertini1997optimal}
Luca Lambertini.
\newblock Optimal fiscal regime in a spatial duopoly.
\newblock \emph{Journal of Urban Economics}, 41\penalty0 (3):\penalty0
  407--420, 1997.

\bibitem[Lindbeck and Weibull(1987)]{lindbeck1987balanced}
Assar Lindbeck and J{\"o}rgen~W Weibull.
\newblock Balanced-budget redistribution as the outcome of political
  competition.
\newblock \emph{Public choice}, 52\penalty0 (3):\penalty0 273--297, 1987.

\bibitem[Loertscher and Muehlheusser(2011)]{loertscher2011sequential}
Simon Loertscher and Gerd Muehlheusser.
\newblock Sequential location games.
\newblock \emph{The RAND Journal of Economics}, 42\penalty0 (4):\penalty0
  639--663, 2011.

\bibitem[Patty(2002)]{patty2002equivalence}
John~W Patty.
\newblock Equivalence of objectives in two candidate elections.
\newblock \emph{Public Choice}, 112\penalty0 (1):\penalty0 151--166, 2002.

\bibitem[Roemer(2009)]{roemer2009political}
John~E Roemer.
\newblock \emph{Political competition: Theory and applications}.
\newblock Harvard University Press, 2009.

\bibitem[Saporiti(2008)]{saporiti2008existence}
Alejandro Saporiti.
\newblock Existence and uniqueness of nash equilibrium in electoral competition
  games: The hybrid case.
\newblock \emph{Journal of Public Economic Theory}, 10\penalty0 (5):\penalty0
  827--857, 2008.

\bibitem[Shaked(1982)]{shaked1982existence}
Avner Shaked.
\newblock Existence and computation of mixed strategy nash equilibrium for
  3-firms location problem.
\newblock \emph{The Journal of Industrial Economics}, pages 93--96, 1982.

\bibitem[Tremblay and Polasky(2002)]{tremblay2002advertising}
Victor~J Tremblay and Stephen Polasky.
\newblock Advertising with subjective horizontal and vertical product
  differentiation.
\newblock \emph{Review of Industrial Organization}, 20\penalty0 (3):\penalty0
  253--265, 2002.

\bibitem[Villani(2009)]{villani2009optimal}
C{\'e}dric Villani.
\newblock \emph{Optimal transport: old and new}, volume 338.
\newblock Springer, 2009.

\bibitem[Wittman(1973)]{wittman1973parties}
Donald~A Wittman.
\newblock Parties as utility maximizers.
\newblock \emph{The American Political Science Review}, 67\penalty0
  (2):\penalty0 490--498, 1973.

\end{thebibliography}

\end{document}